\documentclass[12pt]{article}
\usepackage{}
\usepackage{mathrsfs}
\usepackage{amsfonts}
\usepackage{amsmath}
\usepackage{amssymb}
\usepackage{natbib}
\usepackage[all]{xy}
\usepackage{graphicx}
\usepackage{subfigure}
\usepackage{url}
\usepackage{fancyhdr}
\usepackage{indentfirst}
\usepackage{enumerate}
\usepackage{dsfont,footmisc}
\usepackage{CJK}
\usepackage{color}
\usepackage[british]{babel}
\usepackage[colorlinks=true,citecolor=blue]{hyperref}
\def\d{\mathrm{d}}
\newcommand{\R}{\mathbb{R}}
\renewcommand{\(}{\left(}
\renewcommand{\)}{\right)}

\newtheorem{theorem}{Theorem}[section]

\newtheorem{corollary}{Corollary}[section]

\newtheorem{definition}{Definition}[section]

\newtheorem{lemma}{Lemma}[section]

\newtheorem{remark}{Remark}[section]

\newcommand{\RV}{\mathcal{RV}}
\newenvironment{proof}[1][Proof]{\noindent \textbf{#1.} }{\  \rule{0.5em}{0.5em}}

\begin{document}
\baselineskip=18pt
\title{Second order asymptotics for discounted aggregate claims of continuous-time renewal risk models with constant interest force\thanks{This work  is supported by the National Social Science Fund of China (24BTJ034), the Provincial Natural Science Research Project of Anhui Colleges (KJ2021A0060, KJ2021A0049, 2022AH050067, 2024AH050037) and the  doctoral research initiation fund of Anhui University.}}
\author{ Bingzhen Geng$^a$, Shijie Wang$^a$ and Yang Yang$^b$ \\
{\small $^a$ School of Big Data and Statistics, Anhui University, Hefei, Anhui, 230601, China;}\\
{\small $^b$ School of Statistics and Mathematics, Nanjing Audit University, Nanjing, Jiangsu, 211815, China}}
\date{}
\maketitle
\begin{minipage}{155mm}
\indent{\bf
Abstract:}~This paper investigates the second order asymptotic expansion for tail probabilities of discounted aggregate claims in continuous-time renewal risk
models with constant interest force. Concretely, two types of continuous-time renewal risk models without and with by-claims are separately discussed. By constructing the asymptotic theory and weighted Kesten-type inequality of randomly weighted sums for second order subexponential random variables, second order asymptotic formulae for these two risk models are firstly built. In comparison of the first order asymptotic formulae, our results are more superior and precise, which are demonstrated by some simple numerical studies.



\end{minipage}

\setcounter{equation}{0}
\section{Introduction}

In consideration of insurance practice, assume that an insurer who sells disaster (such as earthquakes, storms, and severe accidents) insurance policies and makes riskless investments. It is known that continuous-time renewal risk models with constant interest force are usually adopted to describe this type of insurance business. Historically, the continuous-time renewal risk models with constant interest force have been studied by many papers. For instance, \cite{ST1995} derived upper and lower bounds of infinite-time ruin probability by martingale methods in the presence of light-tailed claim sizes. Assuming that the claim-arrival process is a homogeneous Poisson
process and the claim sizes are heavy-tailed, \cite{KS1998} obtained an asymptotic estimation of infinite-time ruin probability. Subsequently, \cite{T2005} achieved a uniformly asymptotic formula for finite-time ruin probability under the assumption that the claim sizes are subexponential and the claim-arrival process is a compound Poisson process. \cite{T2007} get an asymptotic formula of infinite-time ruin probability when the claim sizes belong to extended-regularly-varying tailed class. By means of asymptotic results for randomly weighted sums of subexponential random variables (r.v.s), \cite{HT2008} gained asymptotic estimates for discounted aggregate claims of finite-time and infinite-time ruin probabilities in continuous-time renewal risk models with subexponential claims. In their model, the claim sizes are assumed to be a sequence of independent and identically distributed (i.i.d.) r.v.s and the inter-arrival times form another sequence of i.i.d. nonnegative r.v.s, independent of the claim sizes. Recently,  in order to better fit the actual insurance situation, researchers have been interested in imposing various dependence structures among the claim sizes or the inter-arrival times or between them. For example, \cite{CN2007}, \cite{AB2010}, \cite{LTW2010}, among many recent others.

Moreover, another new research trend is about taking both main and by-claims into account into continuous-time renewal risk models, which have attracted much attention in the past decades. In reality, in the event of natural and man-made disasters such as floods and traffic accidents, other insurance claims are likely to occur after the immediate ones. Such a phenomenon leads to consider the classical risk models with by-claims, which was firstly proposed by \cite{WP1985} in a discrete-time case and by \cite{YG2001} in a continuous-time case. For some recent studies on this field, the reader is referred to \cite{YGN2005}, \cite{WL2012}, \cite{Li2013},  \cite{YL2019}, \cite{GZH2019}, \cite{LCF2021}, \cite{LY2022}, \cite{WYLY2023}, among many others. Particularly, for unidimensional continuous-time renewal risk models, \cite{YL2019} obtained some precise asymptotic expansions for finite-time ruin probability with subexponential main and by-claims. \cite{GZH2019} established a uniformly asymptotic formula for finite-time ruin probability of risk models perturbed by diffusion in the presence of long-tailed and dominatedly-varying claim sizes. \cite{LCF2021} derived asymptotic estimate for finite-time ruin probability of time-dependent risk models with subexponential main and by-claims. Whereas, for bidimensional continuous-time renewal risk models, \cite{LY2022} and \cite{WYLY2023} investigated four types of asymptotic finite-time ruin probabilities with subexponential main and by-claims, respectively.

However, for the literature mentioned above, we remark that all the approximations obtained are based on the first order asymptotic behavior. With the increasing growth of the insurance industry, there is a growing need of the exploration and examination of second (or higher) order estimates for enhancing risk analysis. This requirement stems from the collective expectations of insurers and regulators, who demand greater accuracy in assessing risks. The second order expansions for tail probability of aggregate claims have been widely studied under second order conditions, among which second order subexponentiality (denoted by $\mathscr{S}_2$) and second order regular
variation (denoted by $2\RV$) are the most accepted assumptions. Moreover, the utilization of $\mathscr{S}_2$ and $2\RV$  has found widespread applications in various domains, including applied probability, statistics, risk management, telecommunication networks, and numerous other fields. See \cite{OW1986}, \cite{K1988}, \cite{K1989}, \cite{GP1991}, \cite{DR1996}, \cite{DF2006} for more details.

This paper focuses on the second order subexponential distributions which was proposed by \cite{L2012} and studied by numerous scholars subsequently. For the convenience of presentation, we first recall some related definitions and notations. For any distribution function $F$, denote by $\overline F(x)=1-F(x)$ its tail. Assume that $\overline F(x)>0 $ holds for all $x > 0$. For $t>0$, write $\Delta(t)=(0,t]$,
\begin{eqnarray*}
x+\Delta(t)=(x,x+t]
\end{eqnarray*}
and
\begin{eqnarray*}
F(x+\Delta(t))=F(x,x+t]=F(x+t)-F(x).
\end{eqnarray*}
Hereafter, all limit relationships are according to $x\rightarrow\infty$ unless stated otherwise. For two positive functions $a(\cdot)$ and $b(\cdot)$, write $a(x)\sim b(x)$ if $\lim{a(x)}/{b(x)}=1$. It is known that a distribution $F$ on $[0,+\infty)$ is classified as a distribution of the subexponential class, denoted by $\mathscr{S}$, if
\begin{eqnarray*}
{\overline {F^{2 * }}}(x) \sim 2\overline F(x),
\end{eqnarray*}
where $F^{2*}$ denotes $2$-fold convolution of the distribution $F$ with itself. The class $\mathscr{S}$ was first introduced by \cite{C1964}. For more properties of $\mathscr{S}$ and some related classes, we refer the readers to \cite{EKM1997} and \cite{FKZ2013}. A distribution $F$ on $(-\infty,\infty)$ is classified as a distribution of the local long-tailed class $\mathscr{L}_\Delta$, if, for any $t>0$,
\begin{eqnarray*}
F(x+y+\Delta(t))\sim F(x+\Delta(t))
\end{eqnarray*}
holds uniformly in $y \in [0,1]$. The class $\mathscr{L}_\Delta$ was firstly proposed by \cite{AFK2003}. Below are the definitions of second order subexponential distributions supported on the nonnegative half line and the whole real line which were first proposed in \cite{L2012} and \cite{L2014}, respectively.

\begin{definition}
A distribution $F$ on $[0,\infty)$ with finite mean $\mu_F$ is said to belong to the second order subexponential class, denoted by $F\in\mathscr{S}_{2}$, if $F \in \mathscr{L}_{\Delta}$ and
\begin{eqnarray}
\overline {{F^{2 * }}} (x) - 2\overline F(x) \sim 2{\mu _F}F(x,x + 1].\nonumber
\end{eqnarray}
\end{definition}

\begin{definition}
A distribution $F$ on $(-\infty,\infty)$ with finite mean $\mu_F$ is said to belong to the second order subexponential class, denoted by $\widetilde {{\mathscr S}_2}$, if $F^+ \in \mathscr{S}_{2}$ and $F(x,x+1]$ is almost decreasing, where $F^+(x)=F(x){\bf 1}_{\{x\geq 0\}}$ and ${\bf 1}_E$ is the indicator function of a set E.
\end{definition}

As stated in \cite{L2014}, a significant condition imposed on the second order subexponential distributions supported on $(-\infty,\infty)$ is that the local probabilities of these distributions should be almost decreasing. In addition, the second order subexponential class is quite large, which contains many commonly-used distributions, such as Pareto, Lognormal and Weibull (with parameter between 0 and 1) distributions. Furthermore, it is known that $\widetilde {{\mathscr S}_2} \subset \mathscr S$. See \cite{L2012} for details.

For the continuous-time renewal risk models without constant interest force considered, \cite{L2012} derived a second order precise result in an ordinary renewal risk model under the condition that the equilibrium distribution of claims is second order subexponential. \cite{YWC2022} obtained second order asymptotics for infinite-time ruin probability in a compound renewal risk model. For the continuous-time renewal risk models with constant interest force and by-claims considered, more recently, \cite{Lin2021} achieved a second order asymptotic expansion for an ordinary renewal risk model with by-claims, whereas it is also assumed that the equilibrium distribution of claims belong to the second order subexponential distributions. In the view of practice, it is easy to see that the assumption imposed on the claims themselves but not on the equilibrium distribution of claims is more natural and convenient to verify. Besides, \cite{YLY2022} investigated second order tail behavior for stochastic discounted value of aggregate net losses in a discrete-time risk model.

To the best of our knowledge, up to now, there is no work studying the second order asymptotics for discounted aggregate claims of continuous-time renewal risk models with constant interest force and second order subexponential claim sizes. Inspired by this, this paper is devoted to investigating the second order asymptotic expansions of two types of continuous-time renewal risk models without and with by-claims separately. By constructing the asymptotic theory and weighted Kesten-type inequality of randomly weighted sums for second order subexponential r.v.s, this paper first obtains a uniformly second order asymptotic estimation for discounted aggregate claims (Theorem \ref{the:claim} below) of continuous-time renewal risk models without by-claims. Next, under the assumption that the main and by-claims are second order subexponential and locally weak equivalent, the second order asymptotic estimation for discounted aggregate claims (Theorem \ref{the:by-claim} below) of continuous-time renewal risk models with by-claims is also constructed. It is worth mentioning that our obtained result (Theorem \ref{the:by-claim} below) is different from the one in \cite{Lin2021} since it is assumed that the equilibrium distribution of claims is second order subexponential in his paper.

The rest of this paper is organized as follows. Section \ref{sec:claim} gives the second order asymptotic expansion for discounted aggregate claims of continuous-time renewal risk models without by-claims. Section \ref{sec:by-claim} preforms the second order asymptotic expansion for discounted aggregate claims of continuous-time renewal risk models with by-claims. Section \ref{sec:simulation} shows some simulations to illustrate the superiority and preciseness of our results in comparison of the first order asymptotics. Section \ref{sec:conclusion} concludes this paper. Section \ref{sec:appendix} presents some necessary lemmas and the proofs of main results.

\setcounter{equation}{0}
\section{Second order asymptotic for continuous-time renewal risk models without by-claims}\label{sec:claim}

In this section,  consider a continuous-time renewal risk model with interest force, in which claim sizes $\{X_k,k\geq 1\}$ constitute a sequence of i.i.d. and real-valued r.v.s with common distribution $F$, while their arrival times $\{\tau_k, k\geq 1\}$, independent of $\{X_k,k\geq 1\}$, constitute a renewal counting process $N(t)=\sup\{i\in \mathbb{N}:\tau_{i}\leq t\}$ representing the claim-arrival process with finite renewal function $\lambda (t) = EN(t) = \sum_{i = 1}^\infty {P({\tau _i} \le t)}$. The inter-arrival times $\theta_1=\tau_1,\theta_k=\tau_k-\tau_{k-1}, k=2, 3, \ldots $ constitute another sequence of i.i.d. nonnegative and not-degenerate-at-zero r.v.s. Assume that there is a constant force of interest $r \geq 0$. Then, under these settings, the discounted aggregate claims are expressed as the following stochastic process
\begin{align}\label{eq:Dr}
D_r(t)=\int_{0-}^te^{-rs}\d X(s)=\sum_{k=1}^\infty X_ke^{-r\tau_k}\mathbf{1}_{\{\tau_k\leq t\}}.
\end{align}
 For the convenience of later use, define $\Lambda  = \{ t:\lambda (t) > 0\}  = \{ t:P({\tau _1} \le t) > 0\} $ and  write $\Lambda_T =\Lambda\cap[0,T] $. In addition, we also assume that all the random resources appearing in \eqref{eq:Dr} are mutually independent. For model (2.1), \cite{HT2008} investigated the first order asymptotic of tail probability of the stochastic process \eqref{eq:Dr} in the presence of $F\in\mathscr{S}$. They proved that, for arbitrary fixed $T\in\Lambda$, the relation below holds uniformly for all $t \in \Lambda_T$ that
\begin{align*}
P\(D_r(t)>x\)\sim\int_{0-}^t\overline{F}(xe^{ru}) \lambda (\d u).
\end{align*}

Commonly, the first order asymptotic is crude. Thus, studying second order asymptotic of tail probability for aggregate claims becomes more valuable. For the convenience of our presentation, we use the following notations. Let
\begin{align*}
\varphi_{F;\lambda,\lambda}(x;t)&= \int_{{0^ - }}^t \int_{{0^ - }}^{t - v} \Big({{e^{ - rv}}} F(x{e^{r(u + v)}},(x + 1){e^{r(u + v)}}]+ {{e^{ - r(u + v)}}F(x{e^{rv}},(x + 1){e^{rv}}]} \Big)\lambda (\d u) \lambda (\d v).
\end{align*}

Here comes our first main result.

\begin{theorem}\label{the:claim}
Consider the discounted aggregate claims described in \eqref{eq:Dr}. If $F\in\widetilde {{\mathscr S}_2}$, then, for arbitrary fixed $T\in\Lambda$, it holds uniformly for all $t \in \Lambda_T$ that
\begin{align*}
P\(D_r(t)>x\)=\int_{0-}^t\overline{F}(xe^{ru}) \lambda (\d u)+\mu_F\varphi_{F;\lambda,\lambda}(x;t)+o\(1\)\int_{0-}^tF(xe^{ru},(x+1)e^{ru}]\lambda (\d u).
\end{align*}
\end{theorem}

\begin{remark}
 Clearly, our Theorem \ref{the:claim} is more precise than Theorem 2.1 of \cite{HT2008}. Moreover, by some simple calculations, one can easily verify that the last term $\int_{0-}^tF(xe^{ru},(x+1)e^{ru}]\lambda (\d u)$ in Theorem \ref{the:claim} is also the higher order infinitesimal of the second one $\varphi_{F;\lambda,\lambda}(x;t)$.
 \end{remark}


\begin{remark}
In Theorem \ref{the:claim}, if we further restrict the claim sizes belong to some smaller distribution class, for instance, $F$ has a density belonging to regular variation (see Corollary \ref{cor:claim} below) with some index, then the second order asymptotic formula of the discounted aggregate claims becomes more explicit.
\end{remark}

\begin{definition}
	We say that a measurable function $f$ valued on $[0,\infty)$ is \textit{regularly varying} at infinity with index $\alpha \in \R$ if, for all $t>0$,
	$$
	\lim\limits_{x \rightarrow \infty} \frac{f(tx)}{f(x)} = t^\alpha,
	$$
	which is denoted by $f \in \RV_{\alpha}$. Further, for a distribution function $F$, we say $F \in \RV_{\alpha}$ if $\overline F$ satisfies the above relation replacing $f$. In particular,  if $\alpha = 0$, $F$ is called \textit{slowly varying} (at infinity), which is denoted by $F \in \RV_0$. 
\end{definition}

\begin{corollary}\label{cor:claim}
Consider the discounted aggregate claims described in \eqref{eq:Dr}. Further assume that $\{N(t),t>0\}$ is a Poisson process with intensity $\lambda > 0$. If $F$ has a density $f\in\RV_{-(\alpha+1)}$ with $\alpha>1$, then $F\in\widetilde {{\mathscr S}_2}$ and, for arbitrary fixed $T\in\Lambda$, it holds uniformly for all $t \in \Lambda_T$ that
\begin{align*}
P\(D_r(t)>x\)&=\frac{\lambda(1-e^{-\alpha r t})}{\alpha r}\overline{F}(x)+\frac{\mu_F\lambda^2(1-e^{-rt})(1-e^{-\alpha r t})}{\alpha r^2}f(x)+o\(f(x)\).
\end{align*}
\end{corollary}
\begin{proof} Firstly, if $f\in\RV_{-(\alpha+1)}$ with $\alpha>1$, then the proof for $F\in\widetilde {{\mathscr S}_2}$ has been presented in the proof of Corollary 3.2 in \cite{L2014}. Moreover, it follows from \cite{DF2006} that ${F}\in\RV_{-\alpha}$, which implies that
\begin{align*}
\int_{0-}^t\overline F(xe^{ru})\lambda (\d u)&\sim \frac{\lambda (1-e^{-\alpha r t})}{\alpha r}\overline F(x).
\end{align*}
Next, due to $F(x,x+y]\sim yF(x,x+1]$ locally uniformly for $y\in(-\infty,\infty)$ and $F(x,x+1]\sim f(x)$, note that
\begin{align*}
\int_{0-}^tF(xe^{ru},(x+1)e^{ru}]\lambda (\d u)&\sim \int_{0-}^te^{ru}F(xe^{ru},xe^{ru}+1]\lambda (\d u)\nonumber\\
&\sim  \lambda\int_{0-}^te^{ru}f(xe^{ru}) \d u\sim\frac{\lambda(1-e^{-\alpha r t})}{\alpha r}f(x),
\end{align*}
where in the last step we used $f\in\RV_{-(\alpha+1)}$. Similarly, we have
\begin{align}
\varphi_{F;\lambda,\lambda}(x;t)&\sim \lambda^2\int_{{0^ - }}^t \int_{{0^ - }}^{t - v} \Big(e^{ - rv}e^{ r(u+v)} f(x{e^{r(u + v)}})+ e^{ - r(u + v)}e^{rv}f(x{e^{rv}})\Big) \d u\d v\nonumber\\
&\sim\lambda^2f(x)\int_{{0^ - }}^t \int_{{0^ - }}^{t - v} \Big(e^{ - rv} e^{-\alpha r(u + v)}+ {e^{ - r(u + v)}}e^{-\alpha r v}\Big)\d u \d v\nonumber\\
&=\frac{\lambda^2(1-e^{-rt})(1-e^{-\alpha r t})}{\alpha r^2}f(x).\nonumber
\end{align}
Combining all these relations yields the desired result.
\end{proof}

\setcounter{equation}{0}
\section{Second order asymptotic for continuous-time renewal risk models with by-claims}\label{sec:by-claim}

This section mainly focuses on continuous-time renewal risk models with constant interest force and by-claims. Specifically, for each positive integer $k$, assume that an insurer's $k$th main claim $X_k$ occurring at time $\tau_k$ will be accompanied with a by-claim $Y_k$ occurring at $\tau_k+D_k$, where $D_k$ denotes an uncertain delay time (possibly degenerate at $0$). Let $\{X_k;k \geq 1\}$ and $\{Y_k;k \geq 1\}$   be two sequences of i.i.d. real-valued r.v.s with respective common distributions $F$ and $G$, and $\{D_k;k \geq 1\}$ be another sequence of i.i.d. nonnegative r.v.s with common distributions $H$. Assume that the claim-arrival times of the main claims $\{\tau_k;k \geq 1\}$ constitute a renewal sequence such that the inter-arrival times $\theta_1=\tau_1,\theta_k=\tau_k-\tau_{k-1}, k=2, 3, \cdots $ form a sequence of i.i.d. nonnegative r.v.s, which drive the corresponding renewal counting process $N(t)=\sup\{i\in \mathbb{N}:\tau_{i}\leq t\}$ representing the claim-arrival process with finite renewal function $\lambda (t) = EN(t) = \sum\limits_{i = 1}^\infty {P({\tau _i} \le t)}$. Under this setting, the discounted aggregate claims with by-claims are expressed as the stochastic process
\begin{eqnarray}\label{eq:by-claim}
L_r(t)=\sum\limits_{k = 1}^{N(t)} {{X_k}{e^{ - r{\tau _k}}}}  + \sum\limits_{k = 1}^{\infty} {{Y_k}{e^{ - r({\tau _k} + {D_k})}}{{\bf 1}_{\{ {\tau _k} + {D_k} \le t\} }}},~~ t \geq 0,
\end{eqnarray}
where $r \geq 0$ is the constant interest force as before. In addition, we also assume that all the random resources appearing in \eqref{eq:by-claim} are mutually independent.

Now, we are ready to state our next main result. For the convenience of our presentation, we introduce the following notations. Set $(\lambda * H)(t) =\int_{{0^ - }}^tH(t-s)\lambda(\d s)$,
\begin{align*}
\varphi_0(x;t):=\int_{{0^ - }}^t {\overline F(x{e^{ru}})} \lambda (\d u) + \int_{{0^ - }}^t {\overline G(x{e^{ru}})} (\lambda  * H)(\d u),
\end{align*}
\begin{align*}
\widetilde{\varphi}_{G}(x;t)&:=\int_{{0^ - }}^t {\int_{{0^ - }}^{t - v} {{e^{ - rv}}G(x{e^{r(u + v)}},(x + 1){e^{r(u + v)}}](\lambda  * H)(\d u)} } \lambda (\d v)\nonumber\\
&\quad+ \int_{{0^ - }}^t {\int_{{0^ - }}^{t - v} {{e^{ - rv}}G(x{e^{r(v + s)}},(x + 1){e^{r(v + s)}}]H(\d s)} } \lambda (\d v)\nonumber\\
&\quad+ \int_{{0^ - }}^t {\int_{{0^ - }}^{t - v} {\int_{{0^ - }}^{t - v} {{e^{ - r(u + v)}}G(x{e^{r(v + s)}},(x + 1){e^{r(v + s)}}]} H(\d s)} \lambda (\d u)} \lambda (\d v),
\end{align*}
\begin{align*}
\widetilde{\varphi}_{F}(x;t)&:=\int_{{0^ - }}^t {\int_{{0^ - }}^{t - u} {{e^{ - r(u + v)}}F(x{e^{ru}},(x + 1){e^{ru}}](\lambda  * H)(\d v)} } \lambda (\d u)\nonumber\\
&\quad+\int_{{0^ - }}^t {\int_{{0^ - }}^{t - v} {{e^{ - r(v + s)}}F(x{e^{rv}},(x + 1){e^{rv}}]} H(\d s)} \lambda (\d v)\nonumber\\
&\quad+ \int_{{0^ - }}^t {\int_{{0^ - }}^{t - v} {\int_{{0^ - }}^{t - v} {{e^{ - r(v + s)}}F(x{e^{r(u + v)}},(x + 1){e^{r(u + v)}}]H(\d s)} \lambda (\d u)} \lambda (\d v)},
\end{align*}
and
\begin{align*}
&\quad\varphi_{G;\lambda  * H,\lambda}(x;t)\nonumber\\
&:=\int_{{0^ - }}^t { \int_{{0^ - }}^{t - v} \Big({e^{ - rv}}G(x{e^{r(u + v)}},(x + 1){e^{r(u + v)}}]}+{e^{ - r(u + v)}}G(x{e^{ru}},(x + 1){e^{ru}}] \Big)(\lambda  * H)(\d u) \lambda (\d v).
\end{align*}
 \begin{theorem}\label{the:by-claim}
 Consider the discounted aggregate claims described in \eqref{eq:by-claim}. Assume that $F,G\in\widetilde {{\mathscr S}_2}$ and $G(x,x+1]\asymp F(x,x+1]$, then, for arbitrary fixed $T\in\Lambda$, it holds uniformly for all $t \in \Lambda_T$ that
\begin{align}
P\(L_r(t)> x\)&=\varphi_0(x;t)+\mu_F\(\varphi_{F;\lambda,\lambda}(x;t)+\widetilde{\varphi}_{G}(x;t)\)+\mu_G\(\varphi_{G;\lambda  * H,\lambda}(x;t)+\widetilde{\varphi}_{F}(x;t)\)+o\(\Delta(x;t)\)\nonumber
\end{align}
with
\begin{align*}
\Delta(x;t)&=\int_{{0^ - }}^t {\overline F(x{e^{ru}},(x+1){e^{ru}}]} \lambda (\d u) + \int_{{0^ - }}^t {\overline G(x{e^{ru}},(x+1){e^{ru}}]} (\lambda  * H)(\d u).\nonumber
\end{align*}
\end{theorem}

\begin{remark}
 Indeed, our Theorem \ref{the:by-claim} is more precise than Theorem 2.1 of \cite{YL2019}, regarded as the first order asymptotic, since their paper showed that
 $$P\(L_r(t)> x\)\sim\varphi_0(x;t),$$
  and it can be easily proved that the two terms behind $\varphi_0(x;t)$ are negligible.
 \end{remark}

\begin{remark}
It seems that the second order expansion for discounted aggregate claims of continuous-time renewal risk models with constant interest force and by-claims is too complex, but it can be easily calculated by computers, which can be seen in our simulation studies.
\end{remark}

\begin{remark}
Similarly done as in Section \ref{sec:claim}, if we further restrict the claim sizes to some smaller distribution class, then the second order asymptotic formula of the discounted aggregate claims becomes transparent.
\end{remark}

\begin{corollary}
Consider the discounted aggregate claims described in \eqref{eq:by-claim}. Assume that $\{N(t),t>0\}$ is a Poisson process with intensity $\lambda > 0$. If $F$ and $G$ have respective  densities $f$ and $g$ both in $\RV_{-(\alpha+1)}$ with $\alpha>1$ satisfying $ f(x)\asymp g(x)$ and $H$ is an exponential distribution with mean $\widehat{\lambda}^{-1}$, but $\widehat{\lambda}\neq r$ and $\widehat{\lambda}\neq\alpha r$, then $F,G\in\widetilde {{\mathscr S}_2}$ and, for arbitrary fixed $T\in\Lambda$, it holds uniformly for all $t \in \Lambda_T$ that
\begin{align*}
P\(L_r(t)> x\)&=\frac{\lambda(1-e^{-\alpha r t})}{\alpha r}\Big(\overline F(x)+\overline G(x)\Big)-\frac{\lambda(1-e^{-(\alpha r+\widehat\lambda) t})}{\alpha r+\hat\lambda}\overline G(x)+\mu_F\Big(\zeta_{r,\alpha}^{\lambda}(t)f(x)\nonumber\\
&\quad+\chi_{r;\alpha}^{\lambda,\widehat{\lambda}}(t)g(x)\Big)+\mu_G\Big(\omega_{r;\alpha}^{\lambda,\widehat{\lambda}}(t)g(x)+\pi_{r;\alpha}^{\lambda,\widehat{\lambda}}(t)f(x)\Big)+o(f(x)),
\end{align*}
with
\begin{align*}
\zeta_{r;\alpha}^{\lambda}(t)=\frac{\lambda^2(1-e^{-rt})(1-e^{-\alpha r t})}{\alpha r^2},
\end{align*}
\begin{align*}
\chi_{r;\alpha}^{\lambda,\widehat{\lambda}}(t)&=\frac{\lambda\widehat{\lambda}\Big((\alpha+1)\lambda+\alpha r\Big)}{\alpha r^2(\alpha r+\widehat{\lambda})(\alpha+1)}+\frac{\lambda r\Big((\alpha+1)r+\alpha \widehat{\lambda}\Big)e^{-(\alpha+1)rt}}{\alpha r^2(\alpha+1)(r-\widehat{\lambda})}+\frac{\lambda(\widehat{\lambda}-2\lambda)e^{-(\alpha r+\widehat{\lambda})t}}{(\alpha r+\widehat{\lambda})(r-\widehat{\lambda})}\nonumber\\
&\quad-\frac{\lambda^2\widehat{\lambda}e^{-rt}}{\alpha r^2(\alpha r+\widehat{\lambda})}-\frac{\lambda^2e^{-\alpha r t}}{\alpha r^2},
\end{align*}
\begin{align*}
\omega_{r;\alpha}^{\lambda,\widehat{\lambda}}(t)&=\frac{\lambda^2(1-e^{-rt})(1-e^{-\alpha rt})}{\alpha r^2}+\frac{\lambda^2e^{-(\alpha r+\widehat{\lambda})}(1-e^{-(r-\widehat{\lambda})t})}{(r-\widehat{\lambda})(\alpha r+\widehat{\lambda})}+\frac{\lambda^2e^{-rt}(1-e^{{-(\alpha r+\widehat{\lambda})t}})}{(\alpha r+\widehat{\lambda})(r+\alpha r+\widehat{\lambda})}\nonumber\\
&\quad-\frac{\lambda^2(1-e^{-(\alpha+1)rt})}{(\alpha+1)r(\alpha r+\widehat{\lambda})}-\frac{\lambda^2(1-e^{-rt})}{r(r+\alpha r+\widehat{\lambda})},
\end{align*}
and
\begin{align*}
\pi_{r;\alpha}^{\lambda,\widehat{\lambda}}(t)&=\frac{\lambda\widehat{\lambda}(\lambda+\alpha\lambda+\alpha r)}{\alpha r^2(\widehat{\lambda}+r)}+\frac{\lambda^2e^{-rt}}{\alpha r^2}-\frac{\lambda^2\widehat{\lambda}e^{-\alpha rt}}{\alpha r^2(\widehat{\lambda}+r)}+\frac{\lambda\Big(\lambda(\widehat{\lambda}-\alpha r)+\alpha \widehat{\lambda}r\Big){e^{-(\widehat{\lambda}+r)t}}}{\alpha r(\widehat{\lambda}+r)(\widehat{\lambda}-\alpha r)}\nonumber\\
&\quad+\frac{\lambda\widehat{\lambda}\Big((\alpha+1)(\lambda\widehat{\lambda}-\alpha r)+\lambda r(\widehat{\lambda}-\alpha r)\Big)e^{-(\alpha+1)rt}}{\alpha(\alpha+1)r^2(\widehat{\lambda}+r)(\widehat{\lambda}-\alpha r)}-\frac{\lambda^2e^{-(\widehat{\lambda}+\alpha r+r)t}}{\alpha r(\widehat{\lambda}+r)}.
\end{align*}
\end{corollary}
\begin{proof}
Adopting the same proof ideas used in Corollary \ref{cor:claim}, it suffices to calculate all the integrals appearing in right-hand of the relation in Theorem \ref{eq:by-claim}.
First, the asymptotic results of $\int_{{0^ - }}^t {\overline F(x{e^{ru}})} \lambda (\d u)$ and $\varphi_{F;\lambda,\lambda}(x;t)$ have been obtained in Corollary \ref{cor:claim}. Next, by some computations, it follows that
\begin{align*}
\int_{{0^ - }}^t {\overline G(x{e^{ru}})} (\lambda  * H)(\d u)&=\lambda\int_{{0^ - }}^t {\overline G(x{e^{ru}})} (1-e^{-\widehat\lambda t})\d u\nonumber\\
&\sim\lambda\overline G(x)\int_{{0^ - }}^t e^{-\alpha rt} (1-e^{-\widehat\lambda t})\d u\nonumber\\
&=\lambda\Big(\frac{1-e^{-\alpha r t}}{\alpha r}-\frac{1-e^{-(\alpha r+\widehat\lambda) t}}{\alpha r+\widehat\lambda}\Big)\overline G(x).\nonumber
\end{align*}
Similarly done as above, we have
\begin{align*}
\widetilde{\varphi}_{G}(x;t)&\sim\lambda^2 g(x)\int_{{0^ - }}^t {\int_{{0^ - }}^{t - v}} e^{ - rv}e^{ - \alpha r(u+v)}\big(1-e^{-\widehat\lambda u}\big)\d u\d v+ \lambda\widehat\lambda g(x)\int_{{0^ - }}^t {\int_{{0^ - }}^{t - v}} e^{ - rv}e^{ - \alpha r(u+v)}e^{-\widehat\lambda s}\d s\d v\nonumber\\
&\quad+ \lambda^2\widehat\lambda g(x) \int_{{0^ - }}^t \int_{{0^ - }}^{t - v} \int_{{0^ - }}^{t - v} e^{ - r(u + v)}e^{ - \alpha r(v+s)}e^{-\widehat\lambda s}\d s\d u\d v\nonumber\\
&=\chi_{r;\alpha}^{\lambda,\widehat{\lambda}}(t)g(x),
\end{align*}
\begin{align*}
\varphi_{G;\lambda * H,\lambda}(x;t)&\sim\lambda^2g(x)\int_{{0^ - }}^t \int_{{0^ - }}^{t - v}\Big(e^{ - rv}e^{ - \alpha r(u+v)}+e^{ - r(u+v)}e^{-\alpha r u}\Big)\Big(1-e^{-\widehat\lambda u}\Big) \d u \d v\nonumber\\
&=\omega_{r;\alpha}^{\lambda,\widehat{\lambda}}(t)g(x),
\end{align*}
and
\begin{align*}
\widetilde{\varphi}_{F}(x;t)&\sim\lambda^2 f(x)\int_{{0^ - }}^t {\int_{{0^ - }}^{t - u}} e^{ - r(u+v)}e^{ - \alpha ru}\big(1-e^{-\widehat\lambda v}\big)\d v\d u+ \lambda\widehat\lambda f(x)\int_{{0^ - }}^t {\int_{{0^ - }}^{t - v}} e^{ - r(v+s)}e^{ - \alpha rv}e^{-\widehat\lambda s}\d s\d v\nonumber\\
&\quad+ \lambda^2\widehat\lambda f(x) \int_{{0^ - }}^t \int_{{0^ - }}^{t - v} \int_{{0^ - }}^{t - v} e^{ - r(v + s)}e^{ - \alpha r(u+v)}e^{-\widehat\lambda s}\d s\d u\d v\nonumber\\
&=\pi_{r;\alpha}^{\lambda,\widehat{\lambda}}(t)f(x).
\end{align*}
Thus, this completes the proof.
\end{proof}

 \setcounter{equation}{0}\par
\section{Simulation}\label{sec:simulation}

~~ In this section, some simulation studies are performed to check the accuracy of the obtained theoretical result in Theorem \ref{the:by-claim}. The improvement of our result is illustrated via the crude Monte-Carlo method compared with that on the first order asymptotics. Model specifications for the numerical studies are listed below:

$\bullet$ The main claims $\{X_k; k \geq 1\}$ and the by-claims $\{Y_k; k \geq 1\}$ constitute two sequences of i.i.d. nonnegative r.v.s with the Pareto distribution of the form
\begin{eqnarray}
F(x) = 1 - {\left( {\frac{{{\kappa _F}}}{{x + {\kappa _F}}}} \right)^{{\alpha _F}}}, ~~G(x) = 1 - {\left( {\frac{{{\kappa _G}}}{{x + {\kappa _G}}}} \right)^{{\alpha _G}}}, ~~x \geq 0,\nonumber
\end{eqnarray}
for some $\alpha_F> 2$ and $\alpha_G> 2$; or with the Weibull distribution of the form
\begin{eqnarray}
F(x) = 1 - {e^{ -(x/{\kappa_F})^{\alpha_F}}}, ~~G(x) = 1 - {e^{ -(x/{\kappa_G})^{\alpha_G}}}, ~~x \geq 0,\nonumber
\end{eqnarray}
for some $0 <\alpha_F< 1$ and $0 <\alpha_G< 1$.

$\bullet$ The inter-arrival times of the main claims $\{\theta_k; k \geq 1\}$ is a sequence of i.i.d. nonnegative r.v.s with a common exponential distribution of the form
\begin{eqnarray}
P(\theta_k \leq x) = 1 - {e^{ - \lambda x}},~~x \geq 0,\nonumber
\end{eqnarray}
for some $\lambda > 0$. Then the counting process $\{N(t);t \geq 0\}$  form a integer-valued stochastic process with a Poisson distribution of the form
\begin{eqnarray}
P(N(t)=k) = \frac{{{{(\lambda t)}^k}}}{{k!}}{e^{ - \lambda t}},\quad k=0,1,\ldots.\nonumber
\end{eqnarray}

$\bullet$ The delayed times of claims $\{D_k; k \geq 1\}$ is a sequence of i.i.d. nonnegative r.v.s with a common exponential distribution of the form
\begin{eqnarray}
H(x) = 1 - {e^{ - \widehat\lambda x}},~~x \geq 0,\nonumber
\end{eqnarray}
for some $\widehat\lambda > 0$.

For the simulated value of tail probability for discounted aggregate claims, we first generate five random sequences including arrival times of the main claims $\{\tau_k; k \geq 1\}$, the delayed times $\{D_k; k \geq 1\}$, the main claims $\{X_k; k \geq 1\}$, the delayed claim $\{Y_k; k \geq 1\}$ and the counting process $\{N(t);t \geq 0\}$, with a sample size $n$. For each sample $i = 1, ..., n$, we denote the above five sequences by $\{\tau_k^{(i)}; k \geq 1\}$, $\{D_k^{(i)};k \geq 1\}$, $\{X_k^{(i)}; k \geq 1\}$, $\{Y_k^{(i)}; k \geq 1\}$ and $\{N^{(i)}(t);t \geq 0\}$, respectively, and then calculate the following quantity
\begin{eqnarray}
S^{(i)} = \sum\limits_{k = 1}^{N^{(i)}(t)} \Big({X_k^{(i)}}{e^{ - r{\tau _k^{(i)}}}}+{Y_k^{(i)}}{e^{ - r({\tau _k^{(i)}} + {D_k^{(i)}})}}{\mathbf{1}_{\{ {\tau _k^{(i)}} + {D_k^{(i)}} \le t\} }}\Big),\nonumber
\end{eqnarray}
which represents the estimated value for the net loss caused by the $i$th simulated series of accidents. Repeating the algorithm above $n$ times, the value of cumulative claim tail probability can be estimated by
\begin{eqnarray}
\frac{1}{n}\sum\limits_{i = 1}^n {{\mathbf{1}_{\{ {S^{(i)}} > x\} }}}.\nonumber
\end{eqnarray}

Set the sample size $i = 10^6$ in the Pareto case or the Weibull case. The various parameters are set to $\kappa _F = \kappa _G = 2$, $\alpha _F = \alpha _G = 2.3$, $\lambda = 0.2 $, $\widehat\lambda = 0.2 $, and $r = 0.1$ for the Pareto-distributed claims or $\kappa_F = \kappa_G = 1$, $\alpha _F = \alpha _G = 0.3$, $\lambda = 0.1 $, $\widehat\lambda = 0.1 $, and $r = 0.1$ for the Weibull-distributed claims.

We consider both the first and the second order asymptotic values of tail probability for discounted aggregate claims, where, clearly, the first term of second order expansion is the first order asymptotic value, i.e. $\varphi_0(x;t)$. We plot graphics of the simulated values, the first-order asymptotic estimates and the second-order asymptotic estimates for the two cases, respectively, as shown in Figures (a) and (b) below. Figure (a) shows the Pareto case and Figure (b) shows the the Weibull case.

\begin{figure*}[htbp]
\centering
\subfigure[Pareto case]{
\includegraphics[width=2.8in,height=2.3in]{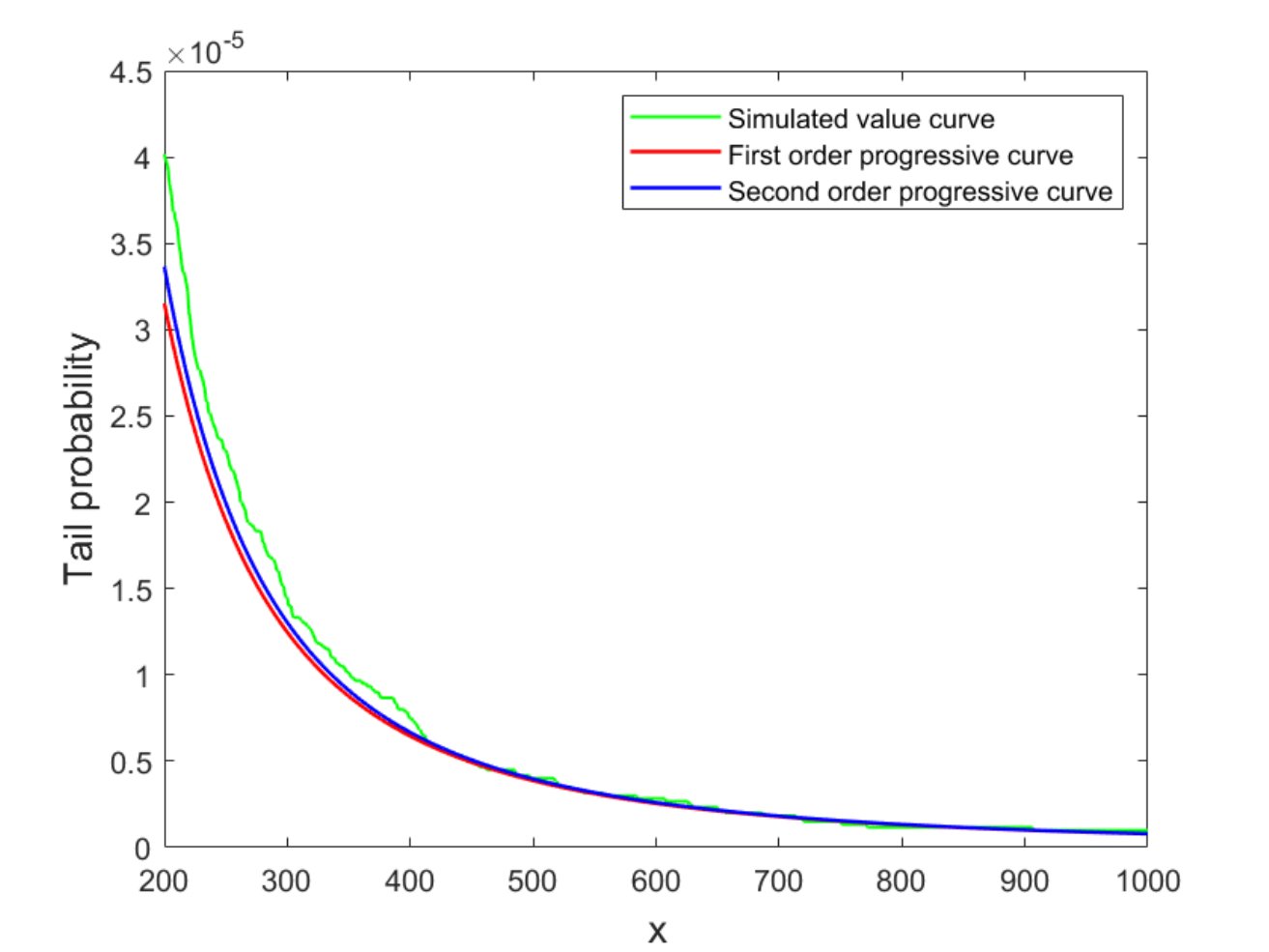}}
\hspace*{0.50cm}
\subfigure[Weibull case]{
\includegraphics[width=2.8in,height=2.3in]{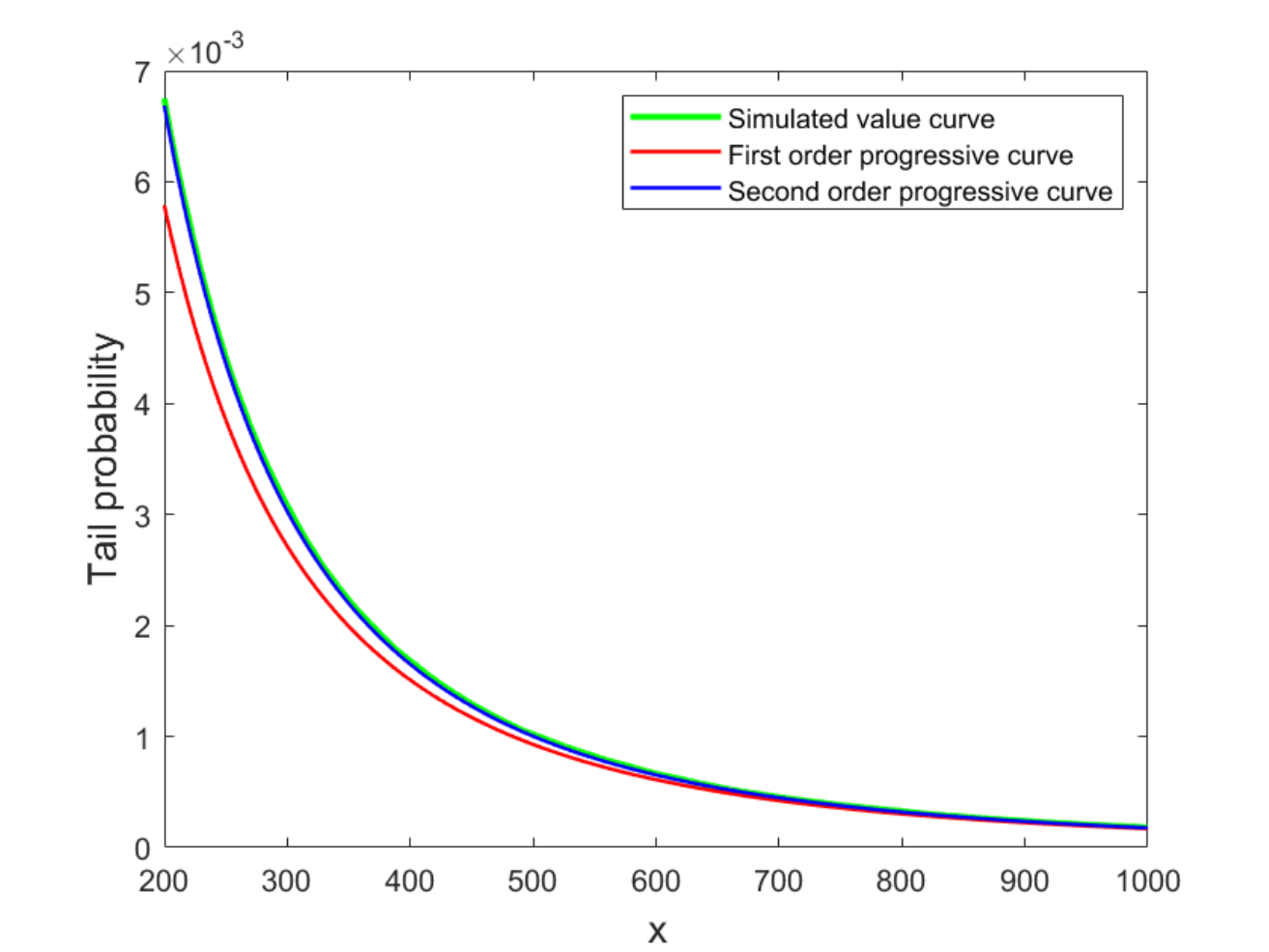}}
\label{fig:2}
\end{figure*}
As shown in Figures (a) and (b), with the increase of $x$ the simulated and the two asymptotic values are closer and decrease gradually. Moreover, the deviations originating from the second order estimate are slightly smaller than the corresponding ones from the first order estimate. This indicates that our second order asymptotic estimate for the cumulative claim tail probability is superior to the traditional first order asymptotic estimate.

\section{Conclusion}\label{sec:conclusion}

The findings from this study mainly make two contributions to the current literature. First, this paper considers continuous-time renewal risk models with constant interest force and without by-claims, in which the claim sizes are assumed to be real-valued second order subexponential r.v.s. By constructing a kesten-type inequality of randomly weighted sums of non-identically distributed second order subexponential r.v.s, the second order asymptotic of discounted aggregate claims is derived. Further, the second order asymptotic of discounted aggregate claims for continuous-time renewal risk models with constant interest force and by-claims is also built. In comparison of the first order asymptotic formulae, our results are more superior and precise, which are demonstrated by some simple numerical studies.

\setcounter{equation}{0}\par
\section{Appendix}\label{sec:appendix}

In this section, we first prepare two lemmas, then the proof of Theorem \ref{the:claim} is given afterwards. Firstly, for simplicity, hereafter, set $\Lambda _{n}^k=\{1,2,\cdots,n\}\setminus\{k\}$ for any $k\in \{1,2,\cdots,n\}$. Lemma \ref{lem:rkesten} is the so-called weighted Kesten-type inequality for second order subexponential distributions, which is of interest in its own right. Lemma \ref{lem:sum} is a non-identically distributed version of Theorem 3.1 of \cite{Lin2020}. Next, we present Lemmas \ref{lem:sum1} and \ref{lem:sum2} for proving Theorem \ref{the:by-claim} and the proof of Theorem \ref{the:by-claim} is given at the end of this paper.

\vspace{3mm}
\begin{lemma}\label{lem:rkesten}
Let $\{X_k,k\geq 1\}$ be a sequence of independent real-valued r.v.s with respective distributions $\{F_k,k\geq 1\}$, and $\{\Theta_k,k\geq 1\}$ be another positive and arbitrarily dependent r.v.s, independent of $\{X_k,k\geq 1\}$. For all $k\geq 1$, denote by $\mu_{F_k}$ the finite expectation of $X_k$. Assume that $F_k\in\mathscr{L}_\Delta $ and $ F_k(x,x+1]\asymp F(x,x+1] $ for some ${F} \in {\widetilde{\mathscr S}_2} $. If, for all $k\geq 1$, $P(\Theta_k \in [a,b])=1$ for some $ 0 < a < b < \infty $, then there exist some constants $A$ and $K=K(\varepsilon,\mu_F,a,b)>0$, irrespective of $n$, such that
 \begin{align*}
    \sup_{x>A}\left|\frac{P\left(\sum\limits_{i=1}^n\Theta_iX_i>x\right)-\sum\limits_{i=1}^nP(\Theta_iX_i>x)}{\sum\limits_{i=1}^nP(\Theta_iX_i\in(x,x+1])}\right|\leq K(1+\varepsilon)^n
 \end{align*}
 holds for all $n\geq1$.
 \end{lemma}
\noindent\textbf{Proof.} To prove this lemma, let $(c_1,\ldots,c_n)\in [a,b]^n$ and for all $n\geq 1$, we first prove that there exist some constants $A$ and $K=K(\varepsilon,\mu_F,a,b)>0$, irrespective of $n$, such that
\begin{align}\label{eq:kesten}
\sup_{(c_1,\ldots,c_n)\in [a,b]^n}\sup_{x>A}\left|\frac{P\left(\sum\limits_{i=1}^nc_iX_i>x\right)-\sum\limits_{i=1}^nP(c_iX_i>x)}{\sum\limits_{i=1}^nP(c_iX_i\in(x,x+1])}\right|\leq K(1+\varepsilon)^n.
\end{align}
Without loss of generality, we can assume that $\varepsilon\in (0,1)$ and $c_{1}=\max\{c_1,\ldots,c_{n+1}\}$. Since $F_k(x,x+1]$ is almost decreasing, thus, by Lemma 4.6 of \cite{Lin2020}, that there exist some constants $A_0, C>0$  such that for all $x>A_0$,
\begin{align}\label{eq:k2}
P(c_{n+1}X_{n+1}\in(x,x+1])&\leq C\frac{ c_1}{c_{n+1}}P(c_1X_{n+1}\in(x,x+1])\nonumber\\
&\asymp C\frac{ c_1}{c_{n+1}}F\left(\frac{x}{c_1},\frac{x+1}{c_1}\right]\nonumber\\
&\asymp C\frac{ c_1}{c_{n+1}}P(c_1X_1\in(x,x+1])\nonumber\\
&:=\tilde{K}P(c_1X_1\in(x,x+1]),
\end{align}
where in the second and third steps we used $ F_k(x,x+1]\asymp F(x,x+1], k\geq 1$. Moreover, for any $i,j\in\{1,\ldots,n\},i\neq j$ and $(c_i,c_j)\in [a,b]^2$, it follows from Lemmas 4.11 and 4.12 of \cite{Lin2020} that there exists sufficiently large constant $A>A_0$ such that for all $x>A$,
\begin{align}\label{eq:k3}
P\(c_iX_i+c_jX_j>x,c_jX_j\leq x-A\)&\leq P(c_iX_i>x)+\frac{\varepsilon}{4(2\tilde{K}+1)}P(c_jX_j\in(x,x+1])\nonumber\\
&\quad+\(b+\frac{\varepsilon}{4(2\tilde{K}+1)}\)\mu_FP(c_iX_i\in(x,x+1]),
\end{align}
\begin{align}\label{eq:k4}
&\quad P\(c_iX_i+c_jX_j\in(x,x+1],c_jX_j\leq x-A\)\nonumber\\
&\leq \(1+\frac{\varepsilon}{4(2\tilde{K}+1)}\)P(c_iX_i\in(x,x+1])+\frac{\varepsilon}{4(2\tilde{K}+1)}P(c_jX_j\in(x,x+1]),
\end{align}
and
\begin{align}\label{eq:k5}
AP\(\sum_{i=1}^{n}b|X_i|>A\)<\frac{\varepsilon}{8}
\end{align}
due to the finiteness of the expectation of $E X_1$.

Now we proceed to use induction to prove \eqref{eq:kesten}. Obviously, \eqref{eq:kesten} is trivial for $n=1$. Now we assume that \eqref{eq:kesten} holds for $n$. Next we aim to prove \eqref{eq:kesten} for $n+1$. To do so, denote by
\begin{align*}
\alpha_{n}=\sup_{(c_1,\ldots,c_n)\in [a,b]^n}\sup_{x>A}\left|\frac{P\(\sum\limits_{i=1}^nc_iX_i>x\)-\sum\limits_{i=1}^nP(c_iX_i>x)}{\sum\limits_{i=1}^nP(c_iX_i\in(x,x+1])}\right|.
\end{align*}
Next, based on the existing sufficiently large $A$, we split the probability $P\(\sum\limits_{i=1}^{n+1}c_iX_i>x\)$ into three parts as:
\begin{align}\label{eq:sum}
&\quad P\(\sum_{i=1}^{n+1}c_iX_i>x,c_{n+1}X_{n+1}\leq x-A\)+P\(\sum_{i=1}^{n+1}c_iX_i>x,c_{n+1}X_{n+1}> x\)\nonumber\\
&\quad+P\(\sum_{i=1}^{n+1}c_iX_i>x,x-A<c_{n+1}X_{n+1}\leq x\)\nonumber\\
&:=J_1(x)+J_2(x)+J_3(x).
\end{align}
For $J_1(x)$, by the definition of $\alpha_n$, \eqref{eq:k3} and \eqref{eq:k4}, it holds that
\begin{align*}
J_1(x)&=\int_{-\infty}^{x-A}P\left(\sum\limits_{i=1}^{n}c_iX_i>x-t\right)P(c_{n+1}X_{n+1}\in \d t)\nonumber\\
&=\int_{-\infty}^{x-A}\left[P\left(\sum\limits_{i=1}^{n}c_iX_i>x-t\right)-\sum_{i=1}^nP(c_iX_i>x-t)\right]P(c_{n+1}X_{n+1}\in \d t)\nonumber\\
&\quad+\sum_{i=1}^nP(c_iX_i+c_{n+1}X_{n+1}>x,c_{n+1}X_{n+1}\leq x-A)\nonumber
\end{align*}
\begin{align*}
&\leq\alpha_n\sum_{i=1}^n\int_{-\infty}^{x-A}P(c_iX_i\in (x-t,x-t+1])P(c_{n+1}X_{n+1}\in \d t)+\sum_{i=1}^n\(P(c_iX_i>x)\right.\nonumber\\
&\quad\left.+b\mu_FP(c_iX_i\in(x,x+1])+\frac{\varepsilon}{4(2\tilde{K}+1)}P(c_{n+1}X_{n+1}\in(x,x+1])\)\nonumber\\
&\leq \alpha_n\sum_{i=1}^n\(\(1+\frac{\varepsilon}{4(2\tilde{K}+1)}\)P(c_iX_i\in(x,x+1])+\frac{\varepsilon}{4(2\tilde{K}+1)}P(c_{n+1}X_{n+1}\in(x,x+1])\)\nonumber\\
&\quad+\sum_{i=1}^n\(P(c_iX_i>x)+\(b+\frac{\varepsilon}{4(2\tilde{K}+1)}\)P(c_iX_i\in(x,x+1])\right.\nonumber\\
&\quad\left.+\frac{\varepsilon}{4(2\tilde{K}+1)}P(c_{n+1}X_{n+1}\in(x,x+1])\)\nonumber\\
&\leq\sum_{i=1}^nP(c_iX_i>x)+\(\(1+\frac{\varepsilon}{4}\)\alpha_n+b\mu_F+\frac{\varepsilon}{4}\)\sum_{i=1}^nP(c_iX_i\in(x,x+1]),
\end{align*}
where in the last step we used the relation \eqref{eq:k2}.  For $J_2(x)$, by Lemmas 4.8 and 4.10 of \cite{Lin2020}, it is easy to check uniformly for $(c_1,\ldots,c_{n+1})\in [a,b]^n$ that,
\begin{align*}
J_2(x)&=P\(\sum_{i=1}^{n}c_iX_i>0,c_{n+1}X_{n+1}> x\)+P\(\sum_{i=1}^{n+1}c_iX_i>x,\sum_{i=1}^{n}c_iX_i\leq 0\)\nonumber\\
&=P\(\sum_{i=1}^{n}c_iX_i>0\)P(c_{n+1}X_{n+1}> x)+\int_{-\infty}^0P\(c_{n+1}X_{n+1}>x-y\)P\(\sum_{i=1}^{n}c_iX_i\in \d y\)\nonumber\\
&=P(c_{n+1}X_{n+1}> x)-\int_{-\infty}^0P\(c_{n+1}X_{n+1}\in (x,x-y]\)P\(\sum_{i=1}^{n}c_iX_i\in \d y\)\nonumber\\
&\leq P(c_{n+1}X_{n+1}> x)+\(\int_{-\infty}^0yP\(\sum_{i=1}^{n}c_iX_i\in \d y\)+\frac{\varepsilon}{8}\)P(c_{n+1}X_{n+1}\in (x,x+1])\nonumber\\
&\leq P(c_{n+1}X_{n+1}> x)+\(\int_{-\infty}^0yP\(\sum_{i=1}^{n}c_iX_i\in \d y\)+\frac{\varepsilon}{8}(1+\varepsilon)\)P(c_{n+1}X_{n+1}\in (x,x+1]).
\end{align*}
For $J_3(x)$, according to the relation \eqref{eq:k5}, we conclude
\begin{align*}
J_3(x)&=P\(\sum_{i=1}^{n}c_iX_i>A,x-A<c_{n+1}X_{n+1}\leq x\)\nonumber\\
&\quad+P\(\sum_{i=1}^{n+1}c_iX_i>x,0<\sum_{i=1}^{n}c_iX_i\leq A,c_{n+1}X_{n+1}\leq x\)\nonumber
\end{align*}
\begin{align*}
&=P\(\sum_{i=1}^{n}c_iX_i>A\)P(c_{n+1}X_{n+1}\in(x-A,x])\nonumber\\
&\quad+\int_0^AP(c_{n+1}X_{n+1}\in(x-y,x])P\(\sum_{i=1}^{n}c_iX_i\in \d y\)\nonumber\\
&\leq \((1+\varepsilon)AP\(\sum_{i=1}^{n}b|X_i|>A\)+\int_0^AyP\(\sum_{i=1}^{n}c_iX_i\in \d y\)\)P(c_{n+1}X_{n+1}\in(x,x+1])\nonumber\\
&\leq\(\int_0^{\infty}yP\(\sum_{i=1}^{n}c_iX_i\in \d y\)+\frac{\varepsilon}{8}(1+\varepsilon)\)P(c_{n+1}X_{n+1}\in (x,x+1]).
\end{align*}
Then, combining the inequalities above into the relation \eqref{eq:sum} yields that
\begin{align*}
P\(\sum_{i=1}^{n+1}c_iX_i>x\)&\leq\sum_{i=1}^{n+1}P(c_iX_i>x)+\(\(1+\frac{\varepsilon}{4}\)\alpha_n+b\mu_F+1\)\sum_{i=1}^nP(c_iX_i\in(x,x+1])\nonumber\\
&\quad+\(nb\mu_F+\frac{\varepsilon}{4}(1+\varepsilon)\)P(c_{n+1}X_{n+1}\in (x,x+1])\nonumber\\
&\leq\sum_{i=1}^{n+1}P(c_iX_i>x)+\(\(1+\frac{\varepsilon}{4}\)\alpha_n+(n+1)b\mu_F\)\sum_{i=1}^{n+1}P(c_iX_i\in(x,x+1]),
\end{align*}
which implies
\begin{align*}
\alpha_n\leq \Big(1+\frac{\varepsilon}{4}\Big)\alpha_{n-1}+nb\mu_F.
\end{align*}
By induction and in view of $\alpha_1=0$, it holds that
\begin{align}\label{eq:kineq}
\alpha_{n}\leq b\mu_F\sum_{i=0}^{n-2}(n-i)\(1+\frac{\varepsilon}{4}\)^i\leq b\mu_Fn^2\(1+\frac{\varepsilon}{4}\)^{n}.
\end{align}
It is easy to see that the right-hand side of \eqref{eq:kineq} does not exceed $K(1+\varepsilon)^n$ for an appropriately chosen constant $K$ and hence, the desired \eqref{eq:kesten} is proven.

Finally, by conditioning on $(\Theta_1,\ldots,\Theta_n)$, we conclude
\begin{align*}
&\quad\left|P\(\sum\limits_{i=1}^n\Theta_iX_i>x\)-\sum\limits_{i=1}^n P\(\Theta_iX_i>x\)\right|\nonumber\\
&\leq\int_a^b\cdots\int_a^b \left|P\(\sum\limits_{i=1}^nc_iX_i>x\)-\sum\limits_{i=1}^nP\(c_iX_i>x\)\right|P\(\Theta_1\in \d c_1,\cdots,\Theta_n\in \d c_n\)\nonumber\\
&\leq\int_a^b\cdots\int_a^bK\(1+\varepsilon\)^{n}\sum\limits_{i=1}^nP\(c_iX_i\in(x,x+1]\)P\(\Theta_1\in \d c_1,\cdots,\Theta_n\in \d c_n\)\nonumber\\
&=K\(1+\varepsilon\)^{n}\sum\limits_{i=1}^n P\(\Theta_iX_i>(x,x+1]\).
\end{align*}
Thus, this completes the proof of Lemma \ref{lem:rkesten}.

\begin{lemma}\label{lem:sum}
Let $\{Z_k,1\leq k\leq n\}$ be $n$ independent real-valued r.v.s with respective distributions ${U_1}, \cdots ,{U_n}$, and $\{\Theta_k,1\leq k\leq n\}$ be another $n$ positive and arbitrarily dependent r.v.s, independent of $\{Z_k,1\leq k\leq n\}$. For all $1 \le k \le n$, denote by $\mu_{Z_k}$ the finite expectation of $Z_k$. Assume that $Z_k\in\mathscr{L}_\Delta $ and $ Z_k(x,x+1]\asymp Z(x,x+1] $ for some ${Z} \in {\widetilde{\mathscr S}_2} $. If there exist two constants $ 0 < a < b < \infty $ such that $P(\Theta_k \in [a,b])=1$ for all $1 \le k \le n$, then it holds that
\begin{align*}
P\(\sum\limits_{k = 1}^n {{\Theta _k}{Z_k}}  > x\)&=\sum\limits_{k = 1}^n {P({\Theta _k}{Z_k} > x)}  + \sum\limits_{k = 1}^n {\sum\limits_{i \in \Lambda _n^k} {{\mu _{{Z_i}}}E({\Theta _i}{\mathbf{1}_{\{ {\Theta _k}{Z_k} \in (x,x + 1]\} }})} }\nonumber\\
&\quad+o\(\sum\limits_{k = 1}^n {P({\Theta _k}{Z_k} \in (x,x + 1])} \).\nonumber
\end{align*}
\end{lemma}
\noindent\textbf{Proof.}
 The proof of Lemma \ref{lem:sum} can be done by going along the same lines of as that of Theorem 3.1 of \cite{Lin2020} by some obvious modifications.  To save space, we omit it here and this ends the proof.

\vspace{2mm}

\noindent\textbf{Proof of Theorem \ref{the:claim}.}
For the convenience of presentation, we write $\Omega_{n}(t) =\{(s_1,\ldots,s_{n+1})\in(0,\infty)^{n+1}: t_i:=\sum_{i=1}^ns_{i}\leq t<\sum_{i=1}^{n+1}s_{i}\}$ for each $t\in\Lambda_T$. For some sufficiently large positive integer $N$, it holds that
\begin{align*}
P\(D_r(t)>x\)&=\sum_{n=1}^\infty\(P\(\sum_{i=1}^nX_ie^{-r\tau_i}\mathbf{1}_{\{N(t)=n\}}>x\)-\sum_{i=1}^nP\(X_ie^{-r\tau_i}\mathbf{1}_{\{N(t)=n\}}>x\)\) \nonumber\\
&\quad +\sum_{n=1}^\infty\sum_{i=1}^nP\(X_ie^{-r\tau_i}\mathbf{1}_{\{N(t)=n\}}>x\)\nonumber\\
&=\(\sum_{n=N+1}^\infty+\sum_{n=1}^N\)\(P\(\sum_{i=1}^nX_ie^{-r\tau_i}\mathbf{1}_{\{N(t)=n\}}>x\)-\sum_{i=1}^nP\(X_ie^{-r\tau_i}\mathbf{1}_{\{N(t)=n\}}>x\)\)\nonumber\\
&\quad +\sum_{i=1}^\infty P\(X_ie^{-r\tau_i}>x,N(t)\geq i\)\nonumber\\
&:=I_1(x;t)+I_2(x;t)+\int_{0-}^t\overline{F}(xe^{ru}) \lambda (\d u).
\end{align*}
To deal with $I_1(x;t)$, by Lemma \ref{lem:rkesten} above and Lemma 4.6 of \cite{Lin2020}, there exist a constant $C>0$ such that uniformly for $t\in\Lambda_T$,
\begin{align*}
I_1(x;t)&\leq \sum_{n=N+1}^\infty K(1+\varepsilon)^n\sum_{i=1}^n\idotsint\limits_{\Omega_{n}(t)}P(X_ie^{-rt_i}\in(x,x+1])\prod_{l=1}^{n+1}G(\d s_l)\nonumber
\end{align*}
\begin{align*}
&=\sum_{n=N+1}^\infty K(1+\varepsilon)^n\sum_{i=1}^n\int_{0-}^t\int_{0-}^{t-u}F(xe^{r(u+v)},(x+1)e^{r(u+v)}]P(N(t-u-v)=n-i)\nonumber\\
&\quad \times P(\tau_i-\tau_1\in \d v)P(\tau_1\in \d u)\nonumber\\
&\leq \sum_{n=N+1}^\infty K(1+\varepsilon)^n\sum_{i=1}^n\int_{0-}^t\int_{0-}^{t-u}Ce^{rv}F(xe^{ru},(x+1)e^{ru}]P(N(t-u-v)=n-i)\nonumber\\
&\quad \times P(\tau_i-\tau_1\in \d v)P(\tau_1\in \d u)\nonumber\\
&\leq Ce^{rt}\sum_{n=N+1}^\infty K(1+\varepsilon)^nn\int_{0-}^tF(xe^{ru},(x+1)e^{ru}]P(N(t-u)=n-1)P(\tau_1\in \d u)\nonumber\\
&\leq C K e^{rt}\int_{0-}^tF(xe^{ru},(x+1)e^{ru}]E\left[(1+\varepsilon)^{N(t-u)}N(t-u)\mathbf{1}_{\{N(t-u)\geq N\}}\right]\lambda (\d u)\nonumber\\
&\leq C K e^{rt}E\left[(1+\varepsilon)^{N(t)}N(t)\mathbf{1}_{\{N(t)\geq N\}}\right]\int_{0-}^tF(xe^{ru},(x+1)e^{ru}]\lambda (\d u)\nonumber\\
&=o\(\int_{0-}^tF(xe^{ru},(x+1)e^{ru}]\lambda (\d u)\),
\end{align*}
where in the last step we used the well-known fact that moment generating function of $N(t)$ is analytic in a neighborhood of 0; see, e.g. \cite{S1946}.
For $I_2(x;t)$, according to Lemma \ref{lem:sum}, it holds uniformly for all $t\in\Lambda_T$ that
\begin{align*}
I_2(x;t)&=\mu_F\sum_{n=1}^N\idotsint\limits_{\Omega_{n}(t)}\sum_{i=1}^n\sum_{k\in\Lambda_n^i}e^{-rt_k}P(X_ie^{-rt_i}\in(x,x+1])\prod_{l=1}^{n+1}G(\d s_l)\nonumber\\
&\quad +o\(\sum_{n=1}^N\sum_{i=1}^n\idotsint\limits_{\Omega_{n}(t)}P(X_ie^{-rt_i}\in(x,x+1])\prod_{l=1}^{n+1}G(\d s_l)\)\nonumber\\
&=\mu_F\(\sum_{n=1}^\infty-\sum_{n=N+1}^\infty\)\sum_{i=1}^n\sum_{k\in\Lambda_n^i}E\left[e^{-r\tau_k}\mathbf{1}_{\{X_ie^{-r\tau_i}\in(x,x+1],N(t)=n\}}\right]\nonumber\\
&\quad +o\(\sum_{n=1}^NP\(X_ie^{-r\tau_i}>x,N(t)\geq i\)\)\nonumber\\
&:=I_{21}(x;t)+I_{22}(x;t)+o\(\int_{0-}^tF(xe^{ru},(x+1)e^{ru}]\lambda (\d u)\).
\end{align*}
For $I_{21}(x;t)$,  by interchanging the orders of the sum, it holds uniformly for all $t\in\Lambda_T$ that
\begin{align*}
I_{21}(x;t)&=\mu_F\sum_{n=1}^\infty\sum_{i=1}^n\(\sum_{k<i}+\sum_{k>i}^n\)E\left[e^{-r\tau_k}\mathbf{1}_{\{X_ie^{-r\tau_i}\in(x,x+1],N(t)=n\}}\right]\nonumber\\
&=\mu_F\sum_{i=1}^\infty\sum_{n=i}^\infty\(\sum_{k<i}\int_{0-}^t\int_{0-}^{t-u}e^{-r(u+v)}F(xe^{ru},(x+1)e^{ru}]P(N(t-u-v)=n-i)\right.\nonumber\\
&\quad \times P(\tau_k\in\d u)P(\tau_i-\tau_k\in \d v)+\sum_{k>i}^n\int_{0-}^t\int_{0-}^{t-u}e^{-ru}F(xe^{r(u+v)},(x+1)e^{r(u+v)}]\nonumber\\
&\quad \left.\times P(N(t-u-v)=n-k)P(\tau_i\in\d u)P(\tau_k-\tau_i\in \d v)\)\nonumber\\
&=\mu_F\int_{0-}^t\int_{0-}^{t-u}e^{-r(u+v)}F(xe^{ru},(x+1)e^{ru}]+e^{-ru}F(xe^{r(u+v)},(x+1)e^{r(u+v)}]\lambda (\d u)\lambda (\d v)\nonumber\\
&=\mu_F\varphi_{F;\lambda,\lambda}(x;t).
\end{align*}
For $I_{22}(x;t)$, by Lemma 4.6 of \cite{Lin2020}, there exist a constant $C>0$, we have that
\begin{align*}
I_{22}(x;t)&\leq\mu_F\sum_{n=N+1}^\infty\sum_{i=1}^{n}(n-1)P(X_ie^{-r\tau_i}\in(x,x+1],N(t)=n)\nonumber\\
&= \mu_F\sum_{n=N+1}^\infty\sum_{i=1}^{n}(n-1)\int_{0-}^t\int_{0-}^{t-u}F(xe^{r(u+v)},(x+1)e^{r(u+v)}]P(N(t-u-v)=n-i)\nonumber\\
&\quad \times P(\tau_1\in\d u)P(\tau_i-\tau_1\in \d v)\nonumber\\
&\leq Ce^{rt}\mu_F\sum_{n=N+1}^\infty n(n-1)\int_{0-}^tF(xe^{ru},(x+1)e^{ru}]P(N(t-u)=n-1) P(\tau_1\in\d u)\nonumber\\
&=o\(\int_{0-}^tF(xe^{ru},(x+1)e^{ru}]\lambda (\d u)\).
\end{align*}
Combining all these results indicates the desired Theorem \ref{the:claim} and this ends the proof.

\vspace{2mm}

\begin{lemma}\label{lem:sum1}
Under the conditions of Theorem \ref{the:by-claim}, we have that
\begin{align}\label{eq:s1}
 \sum\limits_{n = 1}^\infty\sum\limits_{k = 1}^n {\sum\limits_{i = 1}^n}{E\({e^{ - r({\tau _i} + {D_i})}}{\mathbf{1}_{\{ N(t) = n,{\tau _i} + {D_i} \le t\} }}{\mathbf{1}_{\{ {X_k}{e^{ - r{\tau _k}}} \in (x,x + 1]\} }}\)}
=\widetilde{\varphi}_{F}(x;t).
\end{align}
\end{lemma}
\textbf{Proof.} Firstly, by interchanging the order of the summations twice, the left-hand of \eqref{eq:s1} corresponds to
\begin{align*}
&\quad \(\sum\limits_{k = 1}^\infty  \sum\limits_{i = 1}^k \sum\limits_{n = k}^\infty   +  \sum\limits_{k = 1}^\infty  \sum\limits_{i = k + 1}^\infty  \sum\limits_{n = i}^\infty  \) E({e^{ - r({\tau _i} + {D_i})}}{\mathbf{1}_{\{ N(t) = n,{\tau _i} + {D_i} \le t\} }}{\mathbf{1}_{\{ {X_k}{e^{ - r{\tau _k}}} \in (x,x + 1]\} }})\nonumber\\
&:= {L_{1}} + {L_{2}}.
\end{align*}
For $L_{1}$, interchanging the order of the summations again, it holds that
\begin{align*}
{L_{1}}&=\sum\limits_{i = 1}^\infty  {\sum\limits_{k = i}^\infty  {E\big({e^{ - r({\tau _i} + {D_i})}}{\mathbf{1}_{\{ {\tau _k} \le t, {\tau _i} + {D_i} \le t\} }}{\mathbf{1}_{\{ {X_k}{e^{ - r{\tau _k}}} \in (x,x + 1]\} }}\big)} } \nonumber\\
&= \sum\limits_{i = 1}^\infty  {E\big({e^{ - r({\tau _i} + {D_i})}}{\mathbf{1}_{\{ {\tau _i} \le t,{\tau _i} + {D_i} \le t\} }}{\mathbf{1}_{\{ {X_i}{e^{ - r{\tau _i}}} \in (x,x + 1]\} }}\big)}\nonumber\\
&\quad+\sum\limits_{i = 1}^\infty  {\sum\limits_{k = i + 1}^\infty  {E\big({e^{ - r({\tau _i} + {D_i})}}{\mathbf{1}_{\{ {\tau _k} \le t, {\tau _i} + {D_i} \le t\} }}{\mathbf{1}_{\{ {X_k}{e^{ - r{\tau _k}}} \in (x,x + 1]\} }}\big)} } \nonumber\\
&:= {L_{11}} + {L_{12}}.
\end{align*}
We first deal with $L_{11}$. Obviously, conditioning on $(\tau_i,D_i)$ yields that
\begin{align*}
{L_{11}} = \int_{{0^ - }}^t {\int_{{0^ - }}^{t - v} {{e^{ - r(v + s)}}F(x{e^{rv}},(x + 1){e^{rv}}]} H(\d s)} \lambda (\d v).
\end{align*}
In a similar way, for $L_{12}$, by noticing the independence of $\tau_i$ and $\tau_k-\tau_i$ and conditioning on $(\tau_i,\tau_k-\tau_i,D_i)$, it holds that
\begin{align*}
{L_{12}} = \int_{{0^ - }}^t {\int_{{0^ - }}^{t - v} {\int_{{0^ - }}^{t - v} {{e^{ - r(v + s)}}F(x{e^{r(u + v)}},(x + 1){e^{r(u + v)}}]H(\d s)} \lambda (\d u)} \lambda (\d v)}.
\end{align*}
Now we turn to treat $L_{2}$. Similarly done as in the treatment of $L_{1}$, by conditioning on $(\tau_k,\tau_i-\tau_k+D_k)$ and noticing that the identically-distributed property of $D_k$, it holds that
\begin{align*}
{L_{2}} &= \sum\limits_{k = 1}^\infty  {\sum\limits_{i = k + 1}^\infty  {E\big({e^{ - r({\tau _i} - {\tau _k} + {\tau _k} + {D_i})}}{\mathbf{1}_{\{ {\tau _i} - {\tau _k} + {\tau _k} \le t,{\tau _i} - {\tau _k} + {\tau _k} + {D_i} \le t\} }}{\mathbf{1}_{\{ {X_k}{e^{ - r{\tau _k}}} \in (x,x + 1]\} }}\big)} } \nonumber\\
&= \int_{{0^ - }}^t {\int_{{0^ - }}^{t - u} {{e^{ - r(u + v)}}F(x{e^{ru}},(x + 1){e^{ru}}](\lambda  * H)(\d v)} } \lambda (\d u).
\end{align*}
Summarizing the results obtained above yields the desired result and this ends the proof.

\begin{lemma}\label{lem:sum2}
 Under the conditions of Theorem \ref{the:by-claim}, we have that
\begin{align}\label{eq:s2}
\sum\limits_{n = 1}^\infty\sum\limits_{k = 1}^n {\sum\limits_{i = 1}^n}  {E\big({e^{ - r{\tau _i}}}{\mathbf{1}_{\{ N(t) = n,{\tau _k} + {D_k} \le t\} }}{\mathbf{1}_{\{ {Y_k}{e^{ - r{(\tau _k+D_k)}}} \in (x,x + 1]\} }}\big)}
= \widetilde{\varphi}_{G}(x;t).
\end{align}
\end{lemma}
\textbf{Proof.} Since the proof of this lemma is quite similar to the one of Lemma \ref{lem:sum1}, we only state the proof skeletons to save space. Firstly interchanging the order of the summations twice in the left-hand of \eqref{eq:s2} leads to
\begin{align*}
& \(\sum\limits_{k = 1}^\infty  {\sum\limits_{i = 1}^k {\sum\limits_{n = k}^\infty   +  } } \sum\limits_{k = 1}^\infty  {\sum\limits_{i = k + 1}^\infty  {\sum\limits_{n = i}^\infty  } }\) E\({e^{ - r{\tau _i}}}{\mathbf{1}_{\{ N(t) = n,{\tau _k} + {D_k} \le t\} }}{\mathbf{1}_{\{ {Y_k}{e^{ - r({\tau _k} + {D_k})}} \in (x,x + 1]\} }}\)\nonumber\\
:=&~ {K_{1}} + {K_{2}}.
\end{align*}
Similarly to $L_1$, we obtain
\begin{align*}
{K_{1}} &=\sum\limits_{i = 1}^\infty  {E\big({e^{ - r{\tau _i}}}{\mathbf{1}_{\{ {\tau _i} + {D_i} \le t\} }}{\mathbf{1}_{\{ {Y_i}{e^{ - r({\tau _i} + Di)}} \in (x,x + 1]\} }}\big)} \nonumber\\
&\quad+ \sum\limits_{i = 1}^\infty  {\sum\limits_{k = i + 1}^\infty  {E\big({e^{ - r{\tau _i}}}{\mathbf{1}_{\{ {\tau _k} + {D_k} \le t\} }}{\mathbf{1}_{\{ {Y_k}{e^{ - r({\tau _k} + {D_k})}} \in (x,x + 1]\} }}\big)} }\nonumber\\
&= {K_{11}} + {K_{12}}.
\end{align*}
For $K_{11}$, conditioning on $(\tau_i,D_i)$ implies that
\begin{align*}
{K_{11}} =\int_{{0^ - }}^t {\int_{{0^ - }}^{t - v} {{e^{ - rv}}G(x{e^{r(v + s)}},(x + 1){e^{r(v + s)}}]H(\d s)} } \lambda (\d v).
\end{align*}
Similarly, by conditioning on $(\tau_i,\tau_k-\tau_i+D_i)$, we arrive at
\begin{align*}
{K_{12}} &= \sum\limits_{i = 1}^\infty  {\sum\limits_{k = i + 1}^\infty  {E\big({e^{ - r{\tau _i}}}{\mathbf{1}_{\{ {\tau _k} - {\tau _i} + {\tau _i}  + {D_k} \le t\} }}{\mathbf{1}_{\{ {Y_k}{e^{ - r({\tau _k} - {\tau _i} + {\tau _i}  + {D_k})}} \in (x,x + 1]\} }}\big)} }\nonumber\\
&=\int_{{0^ - }}^t {\int_{{0^ - }}^{t - v} {{e^{ - rv}}G(x{e^{r(u + v)}},(x + 1){e^{r(u + v)}}](\lambda  * H)(\d u)} } \lambda (\d v).
\end{align*}
For $K_{2}$, it holds that
\begin{align*}
{K_{2}} &= \sum\limits_{k = 1}^\infty  {\sum\limits_{i = k + 1}^\infty  {E\big({e^{ - r({\tau _i} - {\tau _k} + {\tau _k})}}{\mathbf{1}_{\{ {\tau _i} - {\tau _k} + {\tau _k} \le t,{\tau _k} + {D_k} \le t\} }}{\mathbf{1}_{\{ {Y_k}{e^{ - r({\tau _k} + {D_k})}} \in (x,x + 1]\} }}\big)} } \nonumber\\
&= \int_{{0^ - }}^t {\int_{{0^ - }}^{t - v} {\int_{{0^ - }}^{t - v} {{e^{ - r(u + v)}}G(x{e^{r(v + s)}},(x + 1){e^{r(v + s)}}]} H(\d s)} \lambda (\d u)} \lambda (\d v).
\end{align*}
Thus, the desired result can be achieved and this ends the proof.

\vspace{2mm}

\noindent\textbf{Proof of Theorem \ref{the:by-claim}.} In this proof, we still use the notations introduced in Section \ref{sec:by-claim}. Firstly arbitrarily choosing some large integer $m$ and using the same decomposition method used in Theorem \ref{the:claim}, it follows from Lemma 4.3 of \cite{YL2019} that
\begin{align}\label{eq:Q}
P\(L(t)>x\)&= P\(\sum\limits_{k = 1}^{N(t)} {{X_k}{e^{ - r{\tau _k}}}}  + \sum\limits_{k = 1}^{N(t)} {{Y_k}{e^{ - r({\tau _k} + {D_k})}}{\mathbf{1}_{\{ {\tau _k} + {D_k} \le t\} }}}  > x\)\nonumber\\
&= \sum\limits_{n =  1}^\infty   P\(\sum\limits_{k = 1}^n {({X_k}{e^{ - r{\tau _k}}} + {Y_k}{e^{ - r({\tau _k} + {D_k})}}{\mathbf{1}_{\{ {\tau _k} + {D_k} \le t\} }})}  > x,N(t) = n\)\nonumber
\end{align}
\begin{align}
&=\(\sum\limits_{n =  1}^m+\sum_{n=m+1}^\infty\)\(P\(\sum\limits_{k = 1}^n {({X_k}{e^{ - r{\tau _k}}}{\mathbf{1}_{\{ N(t) = n\} }} + {Y_k}{e^{ - r({\tau _k} + {D_k})}}{\mathbf{1}_{\{ N(t) = n,{\tau _k} + {D_k} \le t\} }})}  > x\)\right.\nonumber\\
&\left.\quad-\sum\limits_{k = 1}^nP\({({X_k}{e^{ - r{\tau _k}}} + {Y_k}{e^{ - r({\tau _k} + {D_k})}}{\mathbf{1}_{\{ {\tau _k} + {D_k} \le t\} }})}  > x,N(t) = n\)\)\nonumber\\
&\quad+\sum\limits_{n =  1}^\infty\sum\limits_{k = 1}^nP\({({X_k}{e^{ - r{\tau _k}}} + {Y_k}{e^{ - r({\tau _k} + {D_k})}}{\mathbf{1}_{\{ {\tau _k} + {D_k} \le t\} }})}  > x,N(t) = n\)\nonumber\\
&:= {Q_1} + {Q_2}+{Q_3}.
\end{align}
For $Q_1$, applying Lemma \ref{lem:sum} with taking $Z_k$ as $X_k$ and ${\Theta _k}$ as $ {e^{ - r{\tau _k}}}{\mathbf{1}_{\{ N(t) = n\} }}$ for $1 \le k \le n$, while taking ${Z_k}$ as
as $ {Y_{k - n}}$ and ${\Theta _k}$ as ${e^{ - r({\tau _{k - n}} + {D_{k - n}})}}{\mathbf{1}_{\{ N(t) = n,{\tau _{k - n}} + {D_{k - n}} \le t\} }}$  for $n + 1 \le k \le 2n$, it yields that
\begin{align}\label{eq:Q1}
{Q_1} &=\sum\limits_{n =  1}^m\(\sum\limits_{k = 1}^{2n} {\sum\limits_{i \in \Lambda _{2n}^k} {\mu _{{Z_i}}}{E({\Theta _i}{\mathbf{1}_{\{ {\Theta _k}{Z_k} \in (x,x + 1]\} }}} )}  + o\(\sum\limits_{k = 1}^n {P({X_k}{e^{ - r{\tau _k}}}{\mathbf{1}_{\{ N(t) = n\} }} \in (x,x + 1])}\right.\right.\nonumber\\
&\quad\left. \left.+\sum\limits_{k = 1}^n {P({Y_k}{e^{ - r({t_k} + {D_k})}}{\mathbf{1}_{\{ N(t) = n,{\tau _k} + {D_k} \le t\} }} \in (x,x + 1])} \)\) \nonumber\\
&:= {Q_{11}} + {Q_{12}} + {Q_{13}}.
\end{align}
For $Q_{12}$ and $Q_{13}$, it is easy to check that, as $x\rightarrow\infty$ first and then $m\rightarrow\infty$,
\begin{align}\label{eq:Q12}
Q_{12}=o\(\int_{0-}^tF(xe^{ru},(x+1)e^{ru}]\lambda (\d u)\)
\end{align}
and
\begin{align}\label{eq:Q13}
Q_{13}=o\(\int_{{0^ - }}^t {\overline G(x{e^{ru}},(x+1){e^{ru}}]} (\lambda  * H)(\d u)\).
\end{align}
To deal with $Q_{11}$, separately partitioning the second and third summation yields that
\begin{align}\label{eq:Q11}
{Q_{11}}&= \sum\limits_{n = 1}^m {\sum\limits_{k = 1}^n {\(\sum\limits_{i \in \Lambda _n^k}  +  \sum\limits_{i = n + 1}^{2n} \) {\mu _{{Z_i}}}E({\Theta _i}{\mathbf{1}_{\{ {\Theta _k}{Z_k} \in (x,x + 1]\} }})} } \nonumber\\
&\quad+\sum\limits_{n = 1}^m {\sum\limits_{k = n + 1}^{2n} {\(\sum\limits_{i = 1}^n  +  \sum\limits_{i \in \Lambda _{2n}^k\setminus \{ 1, \cdots ,n\} } \) {\mu _{{Z_i}}}E({\Theta _i}{\mathbf{1}_{\{ {\Theta _k}{Z_k} \in (x,x + 1]\} }})} } \nonumber\\
&:= {Q_{111}} + {Q_{112}} + {Q_{113}} + {Q_{114}}.
\end{align}
For ${Q_{111}}$, according to the definition of ${Z_k}$ and ${\Theta _k}$ as well as the treatment of $I_{2}(x;t)$ in the proof of Theorem \ref{the:claim}, as $x\rightarrow\infty$ first and then $m\rightarrow\infty$, it holds uniformly for all $t \in \Lambda_T$ that
\begin{align}\label{eq:Q111}
{Q_{111}}=\mu_F\varphi_{F;\lambda,\lambda}(x;t)+o\(\int_{{0^ - }}^t {F(x{e^{rv}},(x + 1){e^{rv}}]} \lambda (\d v)\).
\end{align}
For $Q_{114}$, using a similar treatment as used in $Q_{111}$ with choosing ${Z_k}$ as ${Y_{k - n}}$ and choosing ${\Theta _k}$ as ${e^{ - r({\tau _{k - n}} + {D_{k - n}})}}{\mathbf{1}_{\{ N(t) = n,{\tau _{k - n}} + {D_{k - n}} \le t\} }}$ for $n + 1 \le k \le 2n$, thus, we obtain that, as $x\rightarrow\infty$ first and then $m\rightarrow\infty$,
\begin{align}\label{eq:Q114}
{Q_{114}}= \mu_G\varphi_{G;\lambda*H,\lambda}(x;t)+o\(\int_{{0^ - }}^t {\overline G(x{e^{ru}},(x+1){e^{ru}}]} (\lambda  * H)(\d u)\).
\end{align}
For ${Q_{112}}$, it is obvious that
\begin{align*}
{Q_{112}}&= {\mu _G}\sum\limits_{n = 1}^m {\sum\limits_{k = 1}^n {\sum\limits_{i = n + 1}^{2n} {E\big({e^{ - r({\tau _{i - n}} + {D_{i - n}})}}{\mathbf{1}_{\{ N(t) = n,{\tau _{i - n}} + {D_{i - n}} \le t\} }}{\mathbf{1}_{\{ {X_k}{e^{ - r{\tau _k}}} \in (x,x + 1]\} }}\big)} } } \nonumber\\
& ={\mu _G}\Big(\sum\limits_{n = 1}^\infty  { - \sum\limits_{n = m + 1}^\infty \Big) \sum\limits_{k = 1}^n {\sum\limits_{i = 1}^n {E({e^{ - r({\tau _i} + {D_i})}}{\mathbf{1}_{\{ N(t) = n,{\tau _i} + {D_i} \le t\} }}{\mathbf{1}_{\{ {X_k}{e^{ - r{\tau _k}}} \in (x,x + 1]\} }})} } } \nonumber\\
&:= {Q_{1121}} - {Q_{1122}}.
\end{align*}
It follows from Lemma \ref{lem:sum1} that
\begin{align*}
{Q_{1121}} ={\mu _G}\widetilde{\varphi}_{F}(x;t).
\end{align*}
For $Q_{1122}$, recalling that $\{X_1,\ldots,X_n\}$ and $\{N(t),t\geq0\}$ are independent and conditioning on $\tau_1$ and $(\tau_1,\tau_k-\tau_1)$, respectively, we have that, as $x\rightarrow\infty$ first and then $m\rightarrow\infty$,
\begin{align*}
{Q_{1122}} &\leq \sum\limits_{n = m + 1}^\infty  {\sum\limits_{k = 1}^n n {E\big({e^{ - r{\tau _1}}}{\mathbf{1}_{\{ N(t) = n\} }}{\mathbf{1}_{\{ {X_k}{e^{ - r{\tau _k}}} \in (x,x + 1]\} }}\big)} }  \nonumber\\
 &\leq\sum\limits_{n = m + 1}^\infty  \sum\limits_{k = 1}^n n\int_{{0^ - }}^t {\int_{{0^ - }}^{t - v}{ P(N(t-u-v)= n-k)F(x{e^{r(u + v)}},(x + 1){e^{r(u + v)}}]}}\nonumber\\ &\quad \times P(\tau_k-\tau_1\in \d u)P(\tau_1\in \d v) \nonumber\\
 &\leq \sum\limits_{n = m + 1}^\infty  \sum\limits_{k = 1}^n n\int_{{0^ - }}^t {\int_{{0^ - }}^{t - v}{ P(N(t-u-v)= n-k) c e^{ru}F(x{e^{rv}},(x + 1){e^{r v}}]}}\nonumber\\
 &\quad \times P(\tau_k-\tau_1\in \d u)P(\tau_1\in \d v) \nonumber\\
 &\leq c e^{rt}\sum\limits_{n = m + 1}^\infty n^2\int_{{0^ - }}^t {{ P(N(t-v)= n-1) F(x{e^{rv}},(x + 1){e^{r v}}]}}P(\tau_1\in \d v) \nonumber
 \end{align*}
 \begin{align*}
 &\leq c e^{rt}\int_{{0^ - }}^t {F(x{e^{rv}},(x + 1){e^{rv}}]}E\left[N(t-v)^2\mathbf{1}_{\{N(t-v)\geq m\}}\right] \lambda (\d v)\nonumber\\
 &\leq c e^{rt}E\left[N(t)^2\mathbf{1}_{\{N(t)\geq m\}}\right]\int_{{0^ - }}^t {F(x{e^{rv}},(x + 1){e^{rv}}]} \lambda (\d v)\nonumber\\
 &= o\(\int_{{0^ - }}^t {F(x{e^{rv}},(x + 1){e^{rv}}]} \lambda (\d v)\).
 \end{align*}
Combing the estimations of $Q_{1121}$ and $Q_{1122}$ yields that, as $x\rightarrow\infty$ first and then $m\rightarrow\infty$,
 \begin{align}\label{eq:Q112}
 Q_{112}={\mu _G}\widetilde{\varphi}_{F}(x;t)+o\(\int_{{0^ - }}^t {F(x{e^{rv}},(x + 1){e^{rv}}]} \lambda (\d v)\).
 \end{align}
For ${Q_{113}}$, by the same arguments as in the treatments of ${Q_{112}}$ and Lemma \ref{lem:sum2}, we get that, as $x\rightarrow\infty$ first and then $m\rightarrow\infty$,
\begin{align}\label{eq:Q113}
Q_{113}&=\mu_F\Big(\sum\limits_{n = 1}^\infty  { - \sum\limits_{n = m + 1}^\infty  \Big) \sum\limits_{k = 1}^n {\sum\limits_{i = 1}^n {E({e^{ - r{\tau _i}}}{\mathbf{1}_{\{ N(t) = n,{\tau _k} + {D_k} \le t\} }}{\mathbf{1}_{\{ {Y_k}{e^{ - r({\tau _k} + {D_k})}} \in (x,x + 1]\} }})} } }\nonumber\\
&= {\mu _F}\widetilde{\varphi}_{G}(x;t)+o\(\int_{{0^ - }}^t {\overline G(x{e^{ru}},(x+1){e^{ru}}]} (\lambda  * H)(\d u)\).
\end{align}
 Plugging \eqref{eq:Q111}-\eqref{eq:Q113} into \eqref{eq:Q11} yields that, as $x\rightarrow\infty$ first and then $m\rightarrow\infty$,
\begin{align*}
Q_{11}&=\mu_F\(\varphi_{F;\lambda,\lambda}(x;t)+\widetilde{\varphi}_{G}(x;t)\)+\mu_G\(\varphi_{G;\lambda  * H,\lambda}(x;t)+\widetilde{\varphi}_{F}(x;t)\)+o\(\Delta(x;t)\),
\end{align*}
which also implies that, as $x\rightarrow\infty$ first and then $m\rightarrow\infty$,
\begin{align}\label{eq:Q1*}
Q_{1}=\mu_F\(\varphi_{F;\lambda,\lambda}(x;t)+\widetilde{\varphi}_{G}(x;t)\)+\mu_G\(\varphi_{G;\lambda  * H,\lambda}(x;t)+\widetilde{\varphi}_{F}(x;t)\)+o\(\Delta(x;t)\)
\end{align}
by combining \eqref{eq:Q1}-\eqref{eq:Q13}. Next, to deal with $Q_2$, by Lemma \ref{lem:rkesten} and $I_1(x;t)$ in proof of Theorem \ref{the:claim}, as $x\rightarrow\infty$ first and then $m\rightarrow\infty$, it holds uniformly for all $t\in\Lambda_T$ that
\begin{align}\label{eq:Q2}
Q_2&\leq K\sum_{n=m+1}^\infty(1+\varepsilon)^n\(\sum_{k=1}^{n} P\({X_k}{e^{ - r{\tau _k}}}{\mathbf{1}_{\{ N(t) = n\} }} \in (x,x + 1]\)\right.\nonumber\\
&\quad\quad\left. + \sum_{k=1}^{n}P\({Y_k}{e^{ - r({t_k} + {D_k})}}{\mathbf{1}_{\{ N(t) = n,{\tau _k} + {D_k} \le t\} }} \in (x,x + 1]\)\)\nonumber\\
&=o\(\int_{{0^ - }}^t {F(x{e^{rv}},(x + 1){e^{rv}}]} \lambda (\d v)+\int_{{0^ - }}^t {\overline G(x{e^{ru}},(x+1){e^{ru}}]} (\lambda  * H)(\d u)\).
\end{align}
Finally, by some simple computations, it is clear that
\begin{align}\label{eq:Q3}
Q_{3}=\int_{{0^ - }}^t {\overline F(x{e^{ru}})} \lambda (\d u) + \int_{{0^ - }}^t {\overline G(x{e^{ru}})} (\lambda  * H)(\d u).
\end{align}
 Plugging \eqref{eq:Q1*}-\eqref{eq:Q3} into \eqref{eq:Q} leads to
\begin{align*}
P\(L_r(t)>x\)&=\int_{{0^ - }}^t {\overline F(x{e^{ru}})} \lambda (\d u) + \int_{{0^ - }}^t {\overline G(x{e^{ru}})} (\lambda  * H)(\d u)+\mu_F\(\varphi_{F;\lambda,\lambda}(x;t)+\widetilde{\varphi}_{G}(x;t)\)\nonumber\\
&\quad+\mu_G\(\varphi_{G;\lambda  * H,\lambda}(x;t)+\widetilde{\varphi}_{F}(x;t)\)+o\(\Delta(x;t)\).
\end{align*}
Therefore, the proof is completed.
\appendix


\small
\begin{thebibliography}{}
	
\bibitem[\protect\citeauthoryear{Asimit and Badescu}{2010}]{AB2010}{Asimit, A.V., Badescu, A.L.,} (2010) Extremes on the discounted aggregate claims in a time
dependent risk model. \emph{Scandinavian Actuarial Journal} \text{(2), 93-104}.

\bibitem[\protect\citeauthoryear{Asmussen et al.}{2003}]{AFK2003} {Asmussen, S., Foss, S., Korshunov, D.,} (2003). Asymptotics for sums of random variables with local subexponential behaviour. \emph{Journal of Applied Probability} \text{16(2), 489-518}.

\bibitem[\protect\citeauthoryear{Chistyakov}{1964}]{C1964} {Chistyakov, V. P.,} (1964). A theorem on sums of independent positive random variables and its applications to branching process. \emph{ Theory of Probability and Its Applications} \text{9(4), 640-648}.

\bibitem[\protect\citeauthoryear{Chen and Ng}{2007}]{CN2007} {Chen, Y., Ng, K.W.,} (2007). The ruin probability of the renewal model with constant interest force and negatively dependent heavy-tail claims. \emph{ Insurance: Mathematics and Economics} \text{40, 415-423}

\bibitem[\protect\citeauthoryear{de Haan and Ferreira}{2006}]{DF2006} {de Haan, L., Ferreira, A.,} (2006). Extreme Value Theory: An Introduction. Springer Series in Operations Research and Financial Engineering.  \emph{Springer, New York}.

\bibitem[\protect\citeauthoryear{de Haan and Resnick}{1996}]{DR1996}{de Haan, L., Resnick, S.,} (1996). Second-order regular variation and rates of convergence in extreme-value theory. \emph{ Annals of Probability} \text{97-124}.



\bibitem[\protect\citeauthoryear{Embrechts et al.}{1997}]{EKM1997} {Embrechts, P., Kl\"{u}ppelberg, C., Mikosch, T.,} (1997). Modelling extremal events for insurance and finance. \emph{ Berlin: Springer.}

\bibitem[\protect\citeauthoryear{Foss et al.}{2013}]{FKZ2013} {Foss, S., Korshunov, D., Zachary, S.,} (2013). An introduction to heavy-tailed and subexponential distributions (second edition). \emph{ New York: Springer.}

\bibitem[\protect\citeauthoryear{Gao et al.}{2019}]{GZH2019}{Gao, Q., Zhang, J., Huang, Z.,} (2019). Asymptotics for a delay-claim risk model with diffusion dependence structures and constant force of interest. \emph{ Journal of Computational and Applied Mathematics} \text{353, 219-231}.

\bibitem[\protect\citeauthoryear{Geluk and Pakes}{1991}]{GP1991}{Geluk, J. L.,  Pakes, A. G.,} (1991). Second order subexponential distributions. \emph{Journal of the Australian Mathematical Society} \text{51(1), 73-87}.


\bibitem[\protect\citeauthoryear{Hao and Tang}{2008}]{HT2008} {Hao, X., Tang, Q.,} (2008). A uniform asymptotic estimate for discounted aggregate claims with subexponential tails. \emph{ Insurance: Mathematics and Economics} \text{43(1), 116-120}.

\bibitem[\protect\citeauthoryear{Kl\"{u}ppelberg}{1988}]{K1988}{Kl\"{u}ppelberg, C.,} (1988). Subexponential distributions and integrated tails. \emph{Jounal of
Applied  Probability} \text{25, 132-141}.

\bibitem[\protect\citeauthoryear{Kl\"{u}ppelberg}{1989}]{K1989}{Kl\"{u}ppelberg, C.,} (1989). Subexponential distributions and characterizations of related classes. \emph{Probability Theory and Related Fields} \text{82, 259-269}.


\bibitem[\protect\citeauthoryear{Kl\"{u}ppelberg and Stadtm\"{u}ller}{1998}]{KS1998} {Kl\"{u}ppelberg, C.,  Stadtm\"{u}ller, U.,} (1998). Ruin probabilities in the presence of heavy-tails and interest rates. \emph{Scandinavian Actuarial Journal} \text{1998(1), 49-58}.

\bibitem[\protect\citeauthoryear{Li et al.}{2010}]{LTW2010} {Li, J., Tang, Q., Wu, R.,} (2010). Subexponential tails of discounted aggregate claims in a time-dependent renewal risk model. \emph{ Advances in Applied Probability} \text{42, 1126-1146}.

\bibitem[\protect\citeauthoryear{Li}{2013}]{Li2013} {Li, J.,} (2013). On pairwise quasi-asymptotically independent random variables and their applications.  \emph{ Statistics and Probability Letters} \text{83, 2081-2087}.

\bibitem[\protect\citeauthoryear{Lin}{2012}]{L2012} {Lin, J.,} (2012). Second order asymptotics for ruin probabilities in a renewal risk model with heavy-tailed claims. \emph{ Insurance: Mathematics and Economics}  \text{51(2), 422-429}.

\bibitem[\protect\citeauthoryear{ Lin}{2014}]{L2014}{Lin, J.,} (2014). Second order tail behaviour for heavy-tailed sums and their maxima with applications to ruin theory. \emph{ Extremes} \text{17(2), 247-262}.

\bibitem[\protect\citeauthoryear{Lin}{2019}]{Lin2019} {Lin, J.,} (2019). Second order tail approximation for the maxima of randomly weighted sums with applications to ruin theory and numerical examples. \emph{ Statistics and Probability Letters} \text{153, 37-47}.

\bibitem[\protect\citeauthoryear{ Lin}{2020}]{Lin2020} {Lin, J.,} (2020). Second order tail behaviour of randomly weighted heavy-tailed sums and their maxima. \emph {Communications in Statistics-Theory and Methods} \text{49(11), 2648-2670.}

\bibitem[\protect\citeauthoryear{Lin}{2021}]{Lin2021} {Lin, J.,} (2021). Second order asymptotics for ruin probabilities of the delayed renewal risk model with heavy-tailed claims. \emph {Communications in Statistics-Theory and Methods} \text{50(5), 1200-1209}.

\bibitem[\protect\citeauthoryear{Liu et al.}{2021}]{LCF2021} {Liu, Y., Chen, Z., Fu, K.,} (2021). Asymptotics for a time-dependent renewal risk model with subexponential main claims and by-claims. \emph{ Statistics and Probability Letters} \text{177, 109174}.

\bibitem[\protect\citeauthoryear{Lu and Yuan}{2022}]{LY2022} {Lu, D., Yuan, M.,} (2022). Asymptotic finite-time ruin probabilities for a bidimensional delay-claim risk model with subexponential claims. \emph{ Methodology and Computing in Applied Probability} \text{24, 2265-2286}.

\bibitem[\protect\citeauthoryear{Omey and Willekens}{1986}]{OW1986}{Omey, E.,Willekens, E.,} (1986). Second order behaviour of the tail of a subordinated probability distribution. \emph{Stochastic Processes and Their Applications} \text{21, 339-351}.

\bibitem[\protect\citeauthoryear{Stein}{1946}]{S1946} { Stein, C.,} (1946). A note on cumulative sums. \emph{The Annals of Mathematical Statistics}
\text{17, 498-499}.

\bibitem[\protect\citeauthoryear{Sundt and Teugels}{1995}]{ST1995} {Sundt, B., Teugels, L.,} (1995). Ruin estimates under interest force. \emph{ Insurance: Mathematics and Economics}  \text{16,7-22}.


\bibitem[\protect\citeauthoryear{Tang}{2005}]{T2005} {Tang, Q.,} (2005). The finite-time ruin probability of the compound poisson model with constant interest force. \emph{Journal of Applied Probability} \text{42, 608-619}.

\bibitem[\protect\citeauthoryear{Tang}{2007}]{T2007} {Tang, Q.,} (2007). Heavy tails of discounted aggregrate claims in the continuous-time renewal model. \emph{Journal of Applied Probability} \text{44, 285¨C294}.


\bibitem[\protect\citeauthoryear{Waters and Papatriandafylou}{1985}]{WP1985} {Waters, H.R., Papatriandafylou, A.,} (1985). Ruin probabilities allowing for delay in claims. \emph{ Insurance: Mathematics and Economics} \text{4(2), 113-122}.

\bibitem[\protect\citeauthoryear{Wang et al.}{2023}]{WYLY2023} {Wang, S., Yang, Y., Liu, Y., Yang, L.,} (2023). Asymptotic finite-time ruin probability of a bidimensional renewal risk model with subexponential main claims and delayed claims. \emph{ Methodology and Computing in Applied Probability} \text{25, 76, doi.org/10.1007/s11009-023-10050-1}.

\bibitem[\protect\citeauthoryear{Wu and Li}{2012}]{WL2012} {Wu, X., Li, S.,} (2012). On a discrete time risk model with time-by-claims and a constant dividend barrier. \emph{Insurance Markets and Companies: Analyses and Actuarial Computations} \text{3(1), 50-57}.

\bibitem[\protect\citeauthoryear{Yang and Li}{2019}]{YL2019} {Yang, H., Li, J.,} (2019). On asymptotic finite-time ruin probability of a renewal risk model with subexponential main claims and by-claims. \emph{Statistics and Probability Letters} \text{149, 153-159}.

\bibitem[\protect\citeauthoryear{Yang et al.}{2022}]{YWC2022} {Yang, Y., Wang, X., Chen, S.,} (2022). Second order asymptotics for infinite-time ruin probability in a compound renewal risk model. \emph{Methodology and Computing in Applied Probability} \text{24, 1221-1236}.

\bibitem[\protect\citeauthoryear{Yang et al.}{2022}]{YLY2022}{Yang, Y., Liu, S.,  Yuen, K. C.,} (2022). Second-order tail behavior for stochastic discounted value of aggregate net losses in a discrete-time risk model. \emph{Journal of Theoretical Probability} \text{35(4), 2600-2621}.


\bibitem[\protect\citeauthoryear{Yuen and Gao}{2001}]{YG2001}{Yuen, K.C., Guo, J.,} (2001). Ruin probabilities for time-correlated claims in the compound binomial model. \emph{ Insurance: Mathematics and Economics} \text{29, 47-57}.

\bibitem[\protect\citeauthoryear{Yuen et al.}{2005}]{YGN2005} {Yuen, K.C., Guo, J., Ng, K.W.,} (2005). On ultimate ruin in a delayed-claims risk model. \emph{Journal of Applied Probability} \text{42(1), 163-174}.




\end{thebibliography}
	
\normalsize
\end{document}